\documentclass[11pt]{article}
\usepackage{amsmath,amssymb,amsfonts,mathrsfs,amsthm,graphicx}
\usepackage[top=30truemm,bottom=30truemm,left=35truemm,right=35truemm]{geometry}

\usepackage{url}

\newtheorem{lem}{Lemma}[section]
\newtheorem{prop}{Proposition}[section]
\newtheorem{thm}{Theorem}[section]

\newtheorem{ass}{Assumption}[section]
\theoremstyle{definition}

\newtheorem{ex}{Example}[section]

\theoremstyle{remark}
\newtheorem{rem}{Remark}[section]
\numberwithin{equation}{section}

\renewcommand{\d}{\, \mathrm{d}}

\begin{document}
\title{
Minimizing Benchmark-Relative Drawdown Duration \\ via Occupation Time Penalization
}
\author{Jun Sekine\thanks{The University of Osaka, Graduate School of Engineering Science, Division of Mathematical Science for Social Systems. Email: sekine@sigmath.es.osaka-u.ac.jp}
\thanks{Jun Sekine was supported by JSPS KAKENHI Grant Number JP23K01450.} \and Marcus Wunsch\thanks{ZHAW School of Management and Law, Institute of Wealth \& Asset Management. Email: wuns@zhaw.ch}}

\maketitle

\begin{abstract}
\noindent
We study a continuous-time portfolio optimization problem in which an
investor is evaluated relative to a non-replicable benchmark and seeks
to control the persistence of benchmark-relative underperformance. We
introduce a benchmark-relative drawdown-duration criterion that
penalizes the expected discounted time spent in unfavorable
benchmark-relative performance states.
Despite the path dependence induced by benchmark-relative drawdowns, we
show that the problem admits a one-dimensional Markovian representation
and derive the associated Hamilton--Jacobi--Bellman equation. We obtain
an explicit projection-based characterization of the optimal feedback
control, establish a verification theorem, and identify geometric
settings under which the associated closed-loop reflected diffusion
admits a unique strong solution.
Our results provide a tractable downside-risk-oriented alternative to
classical benchmark-tracking formulations and reveal a novel projection-based
control structure for benchmark-relative risk management.
\end{abstract}

\section{Introduction}
Benchmark-relative investing plays a central role in modern asset management. Pension funds, mutual funds, hedge funds, and other institutional investors are commonly evaluated relative to benchmark indices rather than in absolute terms. As noted by Litterman and Winkelmann~\cite{litterman_managing_1996}, the primary task faced by many professional fund managers is to outperform a specified benchmark. 
As a consequence, a substantial literature has developed around benchmark-relative portfolio optimization and tracking problems, where investors seek either to outperform a benchmark or to control deviations from it under various market and risk constraints. 
Classical benchmark-tracking formulations typically focus on the magnitude of benchmark-relative deviations. 
Roll~\cite{roll_meanvariance_1992} analyzes tracking error through a mean--variance framework, where the variance of benchmark-relative returns serves as 
the principal measure of active risk. 
Subsequent studies have extended this idea to dynamic settings. 
For example, in a continuous-time setting, Zhao~\cite{zhao_dynamic_2007} studies a benchmark-oriented mean--variance portfolio selection problem with {\it terminal} tracking criteria, while Yao, Zhang, and Zhou~\cite{yao_tracking_2006} consider benchmark-tracking formulations based on {\it running} quadratic deviation penalties between portfolio wealth and benchmark wealth.
Beyond tracking-error formulations, benchmark-relative portfolio optimization has also been studied through goal-attainment and shortfall-control criteria. 
Browne~\cite{browne_risk-constrained_2000, browne_beating_1999} considers active portfolio management problems in which investors seek to maximize the probability of outperforming a benchmark or to control benchmark-relative shortfall risk. 
In a related vein, Al-Aradi and Jaimungal~\cite{al-aradi_outperformance_2018} maximize the expected growth-rate differential relative to a performance benchmark while penalizing deviations from a portfolio to be tracked. 

While these approaches provide natural measures of benchmark-relative
risk, they quantify either the magnitude of benchmark-relative
deviations or the likelihood of benchmark-relative shortfalls, and do
not directly account for the duration of benchmark-relative
underperformance.
In this paper, we consider an investor whose performance is evaluated
relative to a {\it non-replicable} benchmark process and introduce a
benchmark-relative drawdown process defined with respect to the running
maximum of benchmark-relative wealth. Since the benchmark is not assumed
to be attainable by the investor, it serves as an aspirational
performance target. Consequently, our focus is on the persistence of
benchmark-relative underperformance rather than solely on the magnitude
of benchmark-relative deviations or the probability of benchmark
outperformance. More specifically, we seek to minimize the expected
discounted time during which the benchmark-relative drawdown exceeds a
prescribed level.
The resulting criterion may be viewed as a
downside-risk-oriented complement to classical benchmark-tracking formulations. 
From a practical perspective, the duration of
benchmark-relative underperformance can be as important as its
magnitude, particularly when evaluating active managers relative to a benchmark.
Moreover, drawdown-based risk measures have become one of the standard
downside-risk indicators in portfolio management because they capture
cumulative losses from previous peaks rather than isolated return
fluctuations. In the benchmark-relative setting, 
benchmark-relative drawdowns provide a natural extension of this idea 
by capturing cumulative relative losses
from previous performance peaks, thereby motivating duration-based
criteria that penalize prolonged periods of severe underperformance.

The importance of persistence-based risk criteria has been recognized in several related contexts. 
Bayraktar and Young~\cite{bayraktar_optimal_2010} introduce an occupation-time criterion and study optimal investment strategies that minimize the expected time that wealth remains below a prescribed level. 
Related drawdown-based risk criteria have subsequently been investigated by Angoshtari, Bayraktar, and Young~\cite{angoshtari_minimizing_2015}, who study portfolio optimization problems involving drawdown risk. 
More recently, Brinker~\cite{brinker_minimal_2021} and Brinker and Schmidli~\cite{brinker_optimisation_2023} explicitly incorporate drawdown duration into stochastic control problems through occupation-time penalties on unfavorable drawdown regions.
Although closely related in spirit, the drawdown-duration formulations of Brinker~\cite{brinker_minimal_2021} and Brinker and Schmidli~\cite{brinker_optimisation_2023} focus on drawdowns of wealth or surplus relative to their running maxima and are developed primarily in insurance and
surplus-management settings. 
In contrast, our focus is on benchmark-relative underperformance in active portfolio
management.
For other drawdown-related works in financial mathematics, we refer to 
Grossman and Zhou~\cite{grossman_optimal_1993}, Cvitanic and Karatzas~\cite{cvitanic_portfolio_1994}, Sekine~\cite{sekine_long-term_2013}, Cherny and Ob\l oj~\cite{cherny_portfolio_2013}, Elie and Touzi~\cite{elie_optimal_2008}, Roche~\cite{roche_asset_2019}, and Bo et al.~\cite{bo_optimal_2026}, for example, where the authors study utility maximizations under drawdown constraints.

The contributions of this paper are fourfold. 
First, we introduce a new benchmark-relative portfolio optimization problem that explicitly accounts for the duration of benchmark-relative underperformance. 
Second, despite the path dependence induced by the running maximum of benchmark-relative wealth, we show that the problem admits a tractable one-dimensional Markovian representation. 
Third, we derive an explicit projection-based feedback rule, which characterizes the optimal benchmark-relative volatility exposure in closed form.
Finally, we provide several examples illustrating how geometric properties of the admissible set can restore strong well-posedness of the closed-loop reflected diffusion, despite potential discontinuities in the optimal feedback coefficient.

The remainder of the paper is organized as follows. 
Section 2 introduces the market model and benchmark-relative wealth dynamics. Section 3 formulates the optimization problem and its decomposition into auxiliary
control problems. 
Section 4 presents the main theoretical results and studies the well-posedness of the optimal closed-loop system. 
Numerical illustrations and concluding remarks are provided in the final section.
\section{Model Setup}

We consider a continuous-time financial market model on a filtered
probability space
$\left(\Omega, \mathcal{F}, \mathbb{P},
(\mathcal{F}_t)_{t \ge 0}\right)$ satisfying the usual conditions.
\\
Let $W = (W_t)_{t \ge 0}$, $W_t:=(W^1_t,\dots,W^d_t)^\top$ be a $d$-dimensional
$\mathcal{F}_t$-Brownian motion.

\subsection{Financial Market}

We consider a market consisting of $n(\le d)$ risky assets, whose price process
$S = (S_t)_{t \ge 0}$, 
$S_t:=(S^1_t,\dots,S^n_t)^\top$
satisfies
\[
\d S_t = \mathrm{diag}(S_t)\left(\sigma \d W_t + \mu \d t\right),
\quad S_0 \in \mathbb{R}_{++}^n,
\]
where 
$\sigma\in {\mathbb R}^{n\times d}$,
$\mu\in {\mathbb R}^n$, and 
$\text{diag}(x) \in {\mathbb R}^{n\times n}$ 
is the diagonal matrix whose $(i,i)$-element is $x^i$ 
for $x:=(x^1,\dots,x^n)\in {\mathbb R}^n$, 
assuming $\sigma\sigma^\top>0$.
A self-financing investor allocates wealth according to a portfolio
process $\pi = (\pi_t)_{t \ge 0}$, where
$\pi_t = (\pi_t^1,\dots,\pi_t^n)^\top$ denotes the fraction of wealth
invested in each risky asset. The corresponding wealth process
$\mathcal{W} = (\mathcal{W}_t)_{t \ge 0}$ satisfies
\[
\d {\mathcal W}_t ={\mathcal W}_t
\left\{
\sum_{j=1}^n \pi^j_t \frac{\d S^j_t}{S^j_t}
+ \left( 1- \sum_{j=1}^n \pi^j_t\right)r\d t
\right\},
\quad
{\mathcal W}_0=w_0\in {\mathbb R}_{++},
\]
where $r \in \mathbb{R}$ is the risk-free rate. 
We see that
\begin{align*}
 d{\mathcal W}_t &={\mathcal W}_t \left[
\pi^\top_t \left\{
\sigma \d W_t + (\mu-r{\bf 1}_n)\d t
\right\}
+ r \d t
\right] \\
&= {\mathcal W}_t \left[
\pi^\top_t \sigma \d W_t 
+\left(
r+\pi^\top \sigma\lambda 
\right)\d t
\right]
\end{align*}
where we denote
${\bf 1}_n:=(1,\dots,1)^\top\in {\mathbb R}^n$
and
\begin{equation}
\label{eq:def-lambda}
\lambda :=
\sigma^\top \left(\sigma \sigma^\top\right)^{-1}(\mu - r\mathbf{1}_n)
\end{equation}
denotes the market price of risk.

\subsection{Benchmark and Relative Wealth}

Let $B = (B_t)_{t \ge 0}$ denote a benchmark process given by
\[
\d B_t = B_t \left( \sigma_B^\top \d W_t + \mu_B \d t \right),
\quad B_0 > 0,
\]
where $\sigma_B \in \mathbb{R}^d$ and $\mu_B \in \mathbb{R}$.
We define the benchmark-relative wealth process by
\[
Y_t := \frac{\mathcal{W}_t}{B_t}.
\]
A standard application of It\^o's formula yields
\begin{align*}
\frac{\d Y_t}{Y_t}
&=\frac{\d{\mathcal W}_t}{{\mathcal W}_t}-\frac{\d B_t}{B_t}
-\frac{\d\langle {\mathcal W},B\rangle_t}{{\mathcal W}_tB_t}
+\frac{\d\langle B\rangle_t}{B_t^2} \\
&= (\sigma^\top \pi_t - \sigma_B)^\top \d W_t
+ \left\{
r - \mu_B + |\sigma_B|^2
+ \pi_t^\top \sigma (\lambda - \sigma_B)
\right\} \d t.
\end{align*}

\subsection{Logarithmic Transformation and Control Variable}

Define the logarithmic relative wealth
\[
L_t := \log Y_t.
\]
Then $L = (L_t)_{t \ge 0}$ satisfies
\begin{align}
 \d L_t &= \frac{\d Y_t}{Y_t}-\frac{1}{2}\frac{\d\langle Y \rangle_t}{Y_t^2}
\nonumber  \\
&=\left(\sigma^\top \pi_t  -\sigma_B\right)^\top \d W_t  
\nonumber \\
+&\left\{r-\mu_B +|\sigma_B|^2 +
\pi_t^\top \sigma \left(\lambda -\sigma_B\right)
-\frac{1}{2}\left| \sigma^\top \pi_t -\sigma_B\right|^2
\right\}\d t \nonumber \\
&= \alpha_t^\top \d W_t 
+\left( p + q^\top \alpha_t 
-\frac{1}{2}\left| \alpha_t\right|^2 \right)\d t,
\label{eq:L}
\end{align}
where we define the control variable
\begin{equation}
\label{eq:def-alpha}
\alpha_t := \sigma^\top \pi_t - \sigma_B, 
\end{equation}
and constants
\begin{equation}
\label{eq:def-pq}
p := r - \mu_B + \sigma_B^\top \lambda,
\quad
q := \lambda - \sigma_B.
\end{equation}

\begin{rem}
The control $\alpha_t$ represents the volatility exposure of the portfolio
relative to the benchmark. It is related to the portfolio strategy $\pi_t$
via \eqref{eq:def-alpha}, and thus measures the excess
volatility with respect to the benchmark. Conversely,
\[
\pi_t = \left(\sigma \sigma^\top\right)^{-1}
\sigma (\alpha_t + \sigma_B),
\]
so that there is a one-to-one correspondence between $\pi$ and $\alpha$.
This representation allows us to formulate the control problem directly
in terms of the benchmark-relative dynamics.
\end{rem}

\subsection{Drawdown and State Variable}

Define the running maximum of the relative wealth by
\[
\bar{Y}_t := \max_{0 \le s \le t} Y_s \vee \bar{y}_0,
\]
where $\bar{y}_0 > 0$ is given. The corresponding relative drawdown process is
\[
\mathrm{DD}_t := \frac{\bar{Y}_t}{Y_t}.
\]
We introduce the logarithmic drawdown process
\[
X_t := \log \mathrm{DD}_t = \bar{L}_t - L_t,
\label{eq:X}
\]
where $\bar{L}_t := \log \bar{Y}_t$. 
The process $X_t$ represents the logarithmic deviation of the
benchmark-relative wealth from its running maximum and serves as the
state variable in the control problem. In particular, large values of
$X_t$ correspond to severe relative drawdowns.
Combining \eqref{eq:L} and \eqref{eq:X}, we see that
\begin{equation}
\d X^\alpha_t
=- \alpha_t^\top \d W_t
-\left(
p + q^\top \alpha_t - \frac12 |\alpha_t|^2
\right) \d t
+ \d \bar L^\alpha_t,
\quad X^\alpha_0=x_0 \ge 0
\label{eq:dyn-X}
\end{equation}
with $x_0:=\log \bar{y}_0-\log y_0$, 
where we write $X^\alpha:=X$, $\bar{L}^\alpha:=\bar{L}$, to highlight the dependency with the associated control process $\alpha:=(\alpha_t)_{t\ge 0}$. 
Note that \eqref{eq:dyn-X} is regarded as the reflecting SDE (at the
barrier $x=0$), where the solution 
$\left(X^\alpha,\bar{L}^\alpha\right)$ satisfies
\begin{itemize}
 \item[\rm (i)] $X^\alpha\ge 0$, 
 \item[\rm (ii)] $\bar{L}^\alpha$ is adapted, non-decreasing,  
and 
$\displaystyle  \int_0^t {\bf 1}_{\{ X^\alpha_s>0\}} 
\d \bar{L}^\alpha_s=0$ for  $t\ge 0$.
\end{itemize}

\subsection{Admissible Controls}

Let $K \subset \mathbb{R}^n$ be a closed convex set representing portfolio
constraints: that is, we assume, 
\[
 \pi_t \in K 
\quad {}^\forall t\ge 0.
\]
For the control process $(\alpha_t)_{t\ge 0}$, 
defined by \eqref{eq:def-alpha}, we define the set of admissible controls by
\begin{equation*}
{\mathscr A}_{\tilde{K}} :=
\left\{
(\alpha_t)_{t\ge 0}
\;\middle|\;
\begin{aligned}
&\text{$d$-dim.\ progressively measurable, such that} \\
&\int_0^t |\alpha_u|^2 \, \d u < \infty
\text{ and }
\alpha_t \in \tilde{K}
\text{ for all } t \ge 0
\end{aligned}
\right\}.
\end{equation*}
where we define
\[
 \tilde{K}:=\left\{ \sigma^\top x-\sigma_B\,
\middle|\, x \in K\right\} \ \subset {\mathbb R}^d.
\]
Throughout the paper, we assume the following.
\begin{ass}\label{ass:non-replicability}
The set $K \subset \mathbb{R}^n$ is nonempty, compact, and convex.
Moreover,
\[
 \sigma_B \not\in \sigma^\top K,
\]
which means that the benchmark volatility is not attainable,
i.e., the benchmark cannot be replicated.
\end{ass}
\begin{rem}
Under this assumption, 
it holds that
\[
\inf_{x\in K} \left|\sigma^\top x - \sigma_B\right| > 0.
\]
Further, the set $\tilde{K}$
is also nonempty, compact, and convex, and satisfies
\[
\underline a:=\inf_{a\in\tilde K}|a|>0.
\]
\end{rem}
\begin{ex}[Non-replicable benchmark]
\noindent{(i)} Let $d>n$.
The attainable set $\sigma^\top K$ is contained 
in a lower-dimensional subspace of $\mathbb R^d$. 
Therefore, one can choose $\sigma_B \notin \sigma^\top K$,
which corresponds to a benchmark driven by risk factors that are not
spanned by the traded assets.

\noindent(ii) Even in $d=n$ case, 
realistic portfolio constraints may prevent the investor from
replicating the benchmark. For instance, consider 
\[
K:=\left\{\pi\in\mathbb R^n \,\middle|\, \pi_i\ge 0 \text{ for all }i,
\ \sum_{i=1}^n \pi_i \le \bar{k} \right\},
\]
where $\bar{k}>0$ is a leverage bound. 
In this case, it is again
possible that $\sigma_B \notin \sigma^\top K$.
\end{ex}

\begin{ex}[Growth Optimal Portfolio in a Larger Market]
Suppose, in addition to the $n$ liquidly
tradable assets $(S^1,\dots,S^n)$, 
there are $(d-n)$ additional assets $(S^{n+1},\dots,S^d)$ 
that are not accessible to the investor. 
Suppose that the full vector of asset prices
$\bar S = (S^1,\dots,S^d)^\top$ follows
\[
\d \bar S_t = \mathrm{diag}\left(\bar S_t\right)
\left(\bar\sigma \d W_t + \bar\mu \d t\right),
\]
with $\bar{\sigma}\in {\mathbb R}^{d\times d}$
and $\bar{\mu}\in {\mathbb R}^d$, assuming $\bar\sigma$ is invertible.
In the enlarged market consisting of all $d$ assets, one can construct
the growth optimal portfolio (GOP), which serves as a natural benchmark
and achieves the maximal long-term growth rate. Denote its volatility by
$\sigma_B$.
However, when only the first $n(<d)$ assets are accessible, the GOP is
no longer replicable. In particular, its volatility $\sigma_B$ does not
belong to the attainable set $\sigma^\top K$, and hence
$\sigma_B \notin \sigma^\top K$.
\end{ex}
\section{Optimization Problem and Decomposition}

\subsection{Main Optimization Problem}

Let $\ell := \log \frak{d}$ denote the logarithmic drawdown threshold. 
The primary objective of the investor is to minimize the expected discounted time during which the drawdown exceeds the prescribed level. 
This leads to the following stochastic control problem:
\begin{align}
\label{eq:MP}
\bar{U}(x) := \inf_{\alpha \in \mathscr{A}_{\tilde K}}
\mathbb{E} \left[
\int_0^\infty e^{-\delta t} \mathbf{1}_{\{X_t > \ell\}} \d t
\,\middle|\, X_0 = x
\right], \quad x \in [0,\infty),
\end{align}
where $\delta>0$ is the discounting rate. 
This formulation penalizes the persistence of large drawdowns and provides a duration-based measure of downside risk. In contrast to classical tracking problems, which focus on instantaneous deviations from a benchmark, the objective above captures the time spent in unfavorable states.

\subsection{Decomposition into Auxiliary Problems}

To analyze problem~\eqref{eq:MP}, we exploit its structural decomposition into two auxiliary control problems corresponding to small and large drawdown regions.

\medskip

\noindent
\textbf{(S) Small drawdown region $[0,\ell]$.}
For $x \in [0,\ell]$, we consider the problem of maximizing the time until
the drawdown exceeds the threshold:
\begin{equation}\label{eq:SDR}
\bar{V}(x) := \inf_{\alpha \in \mathscr{A}_{\tilde K}}
\mathbb{E} \left[
e^{-\delta \tau_d}
\,\middle|\, X_0 = x
\right],
\quad\text{where}\quad
\tau_d := \inf \{ t \ge 0 \mid X_t > \ell \}.
\end{equation}
This problem characterizes the behavior of the system within the small
drawdown region.

\medskip

\noindent
\textbf{(L) Large drawdown region $(\ell,\infty)$.}
For $x > \ell$, we consider the recovery problem from a large drawdown:
\begin{equation}
\bar{W}(x) := \sup_{\alpha \in \mathscr{A}_{\tilde K}}
\mathbb{E} \left[
e^{-\delta \rho_d}
\,\middle|\, X_0 = x
\right],
\quad\text{where}\quad
\rho_d := \inf \{ t \ge 0 \mid X_t \le \ell \}.
\label{eq:LDR}
\end{equation}
This problem describes the optimal recovery dynamics from the large
drawdown region.

\begin{rem}
The decomposition transforms the original problem, which involves a
discontinuous running cost, into two control problems whose payoffs are
determined by hitting times of the threshold $\ell$. This leads to
tractable HJB equations in each region and a free-boundary-type matching
at the threshold.
\end{rem}

\begin{rem}
In the (large) drawdown region, the reflection term in the dynamics \eqref{eq:dyn-X} of $X^\alpha:=X$ does not increase as
\[
\int_0^t 1_{\{X^\alpha_s>0\}}\,\d\bar L^\alpha_s=0.
\]
Hence, during the time that the process remains above the level $\ell(>0)$, the state process evolves according to the controlled diffusion
\[
\d X^\alpha_t=-\alpha_t^\top \d W_t
-\left(p+q^\top \alpha_t-\frac12|\alpha_t|^2\right)\d t. 
\]
\end{rem}
\noindent 
In the subsequent sections, we first analyze two auxiliary problems and then combine the results to obtain a complete characterization of the original optimization problem.
\section{Main Results}
Noting that the running cost term of \eqref{eq:MP} is equal to zero in the region $\{ x<\ell \}$, and equal to one in the region $\{ x>\ell\}$, we distinguish the two regions

\begin{itemize}
\item[(S)] $0\le x <\ell$ (the small drawdown region)
\item[(L)] $x > \ell$ (the large drawdown region)
\end{itemize}
and derive the Hamilton-Jacobi-Bellman equation
for \eqref{eq:MP} in each region (S) and (L):
\begin{align}
\label{eq:HJB-U-left}
-\delta U(x)
+\inf_{a\in\tilde K}
\left[
\frac{|a|^2}{2}U''(x)
-\left(p+q^\top a-\frac{|a|^2}{2}\right)U'(x)
\right]
&=0,
\quad (0\le x<\ell), \\
\label{eq:HJB-U-right}
-\delta U(x)
+1
+\inf_{a\in\tilde K}
\left[
\frac{|a|^2}{2}U''(x)
-\left(p+q^\top a-\frac{|a|^2}{2}\right)U'(x)
\right]
&=0,
\quad (x>\ell), \\
U'(0+)=0, 
\quad
\lim_{x\to\infty}U(x) = \frac{1}{\delta},
\quad
U(\ell-)=U(\ell+)
\quad\text{and}\quad
U'&(\ell-)=U'(\ell+). 
\label{eq:HJB-U-b}
\end{align}
\begin{rem}
The boundary conditions \eqref{eq:HJB-U-b} are explained as follows: 
(i) $U'(0+)=0$ follows from the reflection at $0$.
(ii) As $x\uparrow \infty$, we infer from \eqref{eq:MP} that
\[
\lim_{x\to\infty}U(x) =\int_0^\infty e^{-\delta t} \d t = \frac{1}{\delta}.
\]
(iii) $U(\ell-)=U(\ell+)$ is the matching condition at $x=\ell$, so that the value function $U$ is continuous at the threshold $\ell$. 
(iv) This is the smooth fitting condition at $x=\ell$, ensuring that the value function $U$ is $C^1$ at the threshold: $U'(\ell-)=U'(\ell+)$. 
\end{rem}
\noindent
Also, we write the HJB equations 
for the auxiliary problems \eqref{eq:SDR} and \eqref{eq:LDR} as follows:
\begin{itemize}
 \item[(S)] {\bf Small drawdown region} ($0\le x\le \ell$)
\begin{equation}
\label{eq:HJB-V}
\begin{split}
&-\delta V 
+\inf_{a\in \tilde{K}}
\left[ \frac{|a|^2}{2} V'' 
-\left(
p+q^\top a -\frac{|a|^2}{2}
\right) V'\right]=0, 
\quad (0\le x< \ell), \\
&V(\ell)=1, \quad V'(0+)=0.
\end{split}
\end{equation}
 \item[(L)] {\bf Large drawdown region} ($x>\ell$)
\begin{equation}
\label{eq:HJB-W}
\begin{split}
&-\delta W(x)
+\sup_{a\in\tilde K}
\left[
\frac{|a|^2}{2} W''(x)
-\left(p+q^\top a-\frac12|a|^2\right)W'(x)
\right]
=0,
\quad (x>\ell), \\
&W(\ell)=1,\quad \lim_{x\to\infty} W(x)=0.
\end{split}
\end{equation}
\end{itemize}
\begin{rem}
The boundary conditions of \eqref{eq:HJB-V} and \eqref{eq:HJB-W} 
are interpreted as follows: 
$V(\ell)=1$, $W(\ell)=1$, and $W(\infty)=0$ are naturally seen from 
the definitions of the value functions, \eqref{eq:HJB-V} and \eqref{eq:HJB-W}.
$V'(0+)=0$ follows from the reflection at $0$.
\end{rem}

\subsection{Results for Small Drawdown Region}

First, we introduce the results on the first
auxiliary problem \eqref{eq:SDR} in the small drawdown region.
First, we see the following.
\begin{thm}\label{thm:U-solves-HJB}
For a given $\kappa>0$, there exists a unique solution $V_\kappa\in C^2([0,\ell])$ of the following HJB equation
\begin{equation}
\begin{split}
&-\delta V_\kappa(x)
+\inf_{a\in\tilde K}
\left[
\frac{|a|^2}{2}V_\kappa''(x)
-
\left(
p+q^\top a-\frac{|a|^2}{2}
\right)V_\kappa'(x)
\right]
=0,
\quad x\in[0,\ell], \\
&V_\kappa(0)=\kappa, \quad V'_\kappa(0)=0.
\end{split}
\label{eq:HJB-Vk}
\end{equation}
Moreover, we have the relation that
\[
 V_\kappa(x)=\kappa V_1(x)
\quad (x\in [0,\ell]).
\]
\end{thm}
\begin{proof}
See Appendix A.1.
\end{proof}
\noindent
Using the above Theorem, we obtain the following.
\begin{thm}\label{thm:normalized-V}
{\rm (1)} The solution of the HJB equation \eqref{eq:HJB-V} is expressed as
\begin{equation}
V(x):=\frac{V_\kappa(x)}{V_\kappa(\ell)},
\quad (x\in[0,\ell]),
\label{eq:V-normalized}
\end{equation}
where we use the solution of \eqref{eq:HJB-Vk}.

\smallskip

\noindent{\rm (2)}
For $x\in[0,\ell]$, the following minimizer is uniquely determined as
\begin{align}
a^*(x)
:=&\mathop{\rm argmin}_{a\in\tilde K}
\left[
\frac{|a|^2}{2}V''(x)
-\left(p+q^\top a-\frac{|a|^2}{2}\right)V'(x)
\right] 
=\Pi_{\tilde K}(\beta(x)q).
\label{eq:minimizer}
\end{align}
Here, $\Pi_{\tilde K}$ denotes the metric projection onto
the closed convex set $\tilde K$ and we define 
\begin{equation}
\beta(x):=\frac{V'(x)}{V''(x)+V'(x)},
\label{eq:def-beta-V}
\end{equation}
which satisfies $\beta(x)\ge 0$ for $x\in [0,\ell]$.

\smallskip

\noindent{\rm (3)}
With the function $x\mapsto a^*(x)$ given in \eqref{eq:minimizer}, the reflecting SDE, 
\begin{equation}
\begin{split}
\d{X}^*_t &= -{a}^*\left( X^*_t\right)^\top \d W_t 
-\left\{ p+ q^\top {a}^*\left( X^*_t\right)
-\frac{1}{2}
\left|{a}^*\left( X^*_t\right)\right|^2
\right\}\d t+ \d \overline{L^*_t}, \\
X^*_0 &=x_0\in [0,\ell],
\end{split}
\label{eq:SDE-X}
\end{equation}
has a unique strong solution for $t\in [0,\tau_d]$
that satisfies
\begin{itemize}
 \item[\rm (i)] $X^*\ge 0$, 
 \item[\rm (ii)] ${\overline L^*}$ is adapted, non-decreasing,  
$\overline{L^*_0}=\bar{\ell}_0$,  
and $\displaystyle  \int_0^t {\bf 1}_{\{ X^*_s>0\}} \d \overline{L^*_s}=0$
for  $t\in [0,\tau_d]$.
\end{itemize}
The feedback-form strategy
$(\alpha^*_t)_{t\in [0,\tau_d]}$ defined by
\[
 \alpha^*_t 
:=a^*\left( X^*_t\right)
\quad t\in \left[ 0,\tau_d\right]
\]
is optimal for the problem \eqref{eq:SDR} 
and $X^*\equiv X^{\alpha^*}$ holds.
Moreover, $V\equiv \bar{V}$ holds.
\end{thm}
\begin{proof}
See Appendix A.2.
\end{proof}
\begin{rem}
$L^*:=\overline{L^*}-X^*$ solves 
the (path-dependent) SDE,
\begin{equation}
\begin{split}
\d{L}^*_t &= {a}^*\left( \overline{L^*_t} -L^*_t\right)^\top \d W_t 
+\left\{ p+ q^\top {a}^*\left( \overline{L^*_t} -L^*_t\right)
-\frac{1}{2}
\left|{a}^*\left( \overline{L^*_t} -L^*_t\right)\right|^2
\right\}\d t, \\
L^*_0&=\ell_0:=\bar{\ell}_0-x_0.
\end{split}
\label{eq:SDE-L}
\end{equation}
It holds that
\[
\overline{L^*_t}:=\max_{s\in [0,t]}L^*_s \vee \bar{\ell}_0, 
\]
which is seen from the solution of the associated Skorohod equation
(see Lemma 6.14 in Chapter 3 of Karatzas and Shreve \cite{karatzas_brownian_1998}).
\end{rem}


\subsection{Results for Large Drawdown Region}
Next, we introduce the results on the second
auxiliary problem \eqref{eq:LDR} in the large drawdown region.
For introducing the explicit solution of HJB equation 
\eqref{eq:HJB-W}, we consider the quadratic equation
\begin{equation}
f_a(r):=\frac{|a|^2}{2}r^2+
\left\{ p+q^\top a-\frac{|a|^2}{2}\right\}
\,r-\delta=0
\label{eq:quadratic-root}
\end{equation}
with respect to $r$, 
which has a unique positive root, denoted by $r(a)$.
Explicitly,
\begin{equation}
r(a)=\frac{
-\left(p+q^\top a-\frac{|a|^2}{2}\right)
+\sqrt{\left(p+q^\top a-\frac{|a|^2}{2}\right)^2+2\delta |a|^2}
}{|a|^2}.
\label{eq:r-a}
\end{equation}
Because $\tilde K$ is compact and the map $a\mapsto r(a)$ is continuous, 
there exists
$\bar{r}\in {\mathbb R}>0$ and $a^\star\in\tilde K$ such that
\begin{align}
\label{eq:bar-r}
\bar r:&=\min_{a\in\tilde K}r(a)=r(a^\star), \\
\label{eq:a-star}
a^\star\in& \mathop{\rm argmin}_{a\in \tilde{K}} r(a).
\end{align}
We obtain the following.
\begin{thm}\label{thm:hjb-sol-W}
{\rm (1)} The function 
\begin{equation}
    W(x):=e^{-\bar r(x-\ell)}, \quad x\ge \ell, 
    \label{eq:W-explicit}
\end{equation}
with $\bar r$ given by \eqref{eq:bar-r},
solves the HJB equation \eqref{eq:HJB-W}.

\smallskip

\noindent{\rm (2)}
The constant strategy $\alpha^*:=(\alpha^*_t)_{t\in [0,\rho_d]}$ 
given by 
\[
\alpha^*_t:= a^\star
\]
with \eqref{eq:a-star} is optimal for the recovery problem \eqref{eq:LDR}, 
and $W\equiv \bar{W}$ holds.
\end{thm}

\begin{proof}
See Appendix A.3. 
\end{proof}

\subsection{Results for the Main Optimization Problem}

We now state the main result of the paper, which provides a complete
characterization of the value function and the optimal control for
problem \eqref{eq:MP}. In this part, we additionally impose the following.
\begin{ass}
For $\bar{r}$, defined by \eqref{eq:bar-r}, assume 
\[
\bar r<1,
\]
or equivalently, 
\[
p+q^\top a=
r + \frac12 |\lambda|^2 - \mu_B
+(\lambda - \sigma_B)^\top a>\delta
\quad\text{with ${}^\exists a \in \tilde{K}$.}
\]
where we use \eqref{eq:def-lambda} and \eqref{eq:def-pq}.
\end{ass}
\begin{rem}[Interpretation of the condition $\bar{r}<1$]
Recall that $\bar r = \min_{a\in\tilde K} r(a)$, where $r(a)$ is the positive
root of the quadratic equation $f_a(r)=0$, i.e., \eqref{eq:quadratic-root}.
A direct computation shows that
\begin{align*}
\bar r<1
\quad\Leftrightarrow\quad&
f_a(1)>0\quad\text{with ${}^\exists a\in \tilde{K}$} \\
\quad\Leftrightarrow\quad&
p + q^\top a > \delta
\quad\text{with ${}^\exists a\in \tilde{K}$},
\end{align*}
hence, we get the second expression in Assumption 4.2.
The quantity $p + q^\top a$ represents the ``effective drift'' of the
log-relative wealth process with respect to the benchmark.
Hence, the condition $\bar r<1$ means that the investor can construct a
portfolio whose relative growth rate exceeds the discount rate $\delta$.
\end{rem}
\noindent
Let 
\begin{equation}
\label{eq:kappa-star-C-star}
\kappa^\star:=\frac{1}{\delta}
\left\{V_1(\ell)+\frac{V_1'(\ell-)}{\bar r} \right\}^{-1},
\quad
C^\star:=\frac{\kappa^\star V_{1}'(\ell-)}{\bar r},
\end{equation}
where we use the function $x\mapsto V_\kappa(x)=\kappa V_1(x)$ ($\kappa>0$),
the solution of \eqref{eq:HJB-Vk} 
and $\bar{r}>0$ given by \eqref{eq:bar-r}.
Using these constants, define the function
$U:{\mathbb R}_+ \to {\mathbb R}_+$ by
\begin{equation}
U(x)
:=\begin{cases}
\kappa^\star V_{1}(x) & \text{for $0\le x\le \ell$}, \\
-C^\star e^{-\bar r(x-\ell)}+\dfrac{1}{\delta}& 
\text{for $x>\ell$}.
\end{cases}
\label{eq:def-U}
\end{equation}
We then obtain the following.
\begin{thm}\label{thm:hjb-sol-U}
{\rm (1)} The function \eqref{eq:def-U} satisfies
$U\in C^1([0,\infty)) \cap C^2([0,\ell)\cup (\ell,\infty))$
and solves HJB equation \eqref{eq:HJB-U-left}-\eqref{eq:HJB-U-b}.

\smallskip

\noindent{\rm (2)} 
For $x\in [0,\ell) \cup (\ell,\infty)$, 
the following minimizer is uniquely determined as
\begin{align}
a^*(x)
:&=\mathop{\rm argmin}_{a\in\tilde K}
\left[
\frac{|a|^2}{2}U''(x)
-\left(p+q^\top a-\frac{|a|^2}{2}\right)U'(x)
\right] 
=\Pi_{\tilde K}(\beta(x)q),
\label{eq:minimizer-U}
\end{align}
where 
\begin{equation}
\beta(x):=\frac{U'(x)}{U''(x)+U'(x)}
\label{eq:def-beta-U}
\end{equation}
satisfies $\beta(x)\ge 0$. Moreover, 
$a^*(x)=a^\star$ holds for $x\ge \ell$, where we use \eqref{eq:a-star}.

\smallskip

\noindent{\rm (3)}
We extend the function $x\mapsto a^*(x)$ given 
in \eqref{eq:minimizer-U} 
to $a^*: [0,\infty) \to \tilde{K}$ by setting $a^*(\ell):=a^\star$.
Suppose that the reflecting SDE, 
\begin{equation}
\begin{split}
\d{X}^*_t&= -{a}^*\left( X^*_t\right)^\top \d W_t 
-\left\{ p+ q^\top {a}^*\left( X^*_t\right)
-\frac{1}{2}
\left|{a}^*\left( X^*_t\right)\right|^2
\right\}\d t+ \d \overline{L^*_t}, \\
X^*_0&=x_0\in [0,\ell],
\end{split}
\label{eq:SDE-X-global}
\end{equation}
has a unique strong solution that satisfies
\begin{itemize}
 \item[\rm (i)] $X^*\ge 0$, 
 \item[\rm (ii)] ${\overline L^*}$ is adapted, non-decreasing,  
$\overline{L^*_0}=\bar{\ell}_0$,  
and $\displaystyle  \int_0^t {\bf 1}_{\{ X^*_s>0\}} \d \overline{L^*_s}=0$
for  $t\ge 0$.
\end{itemize}
Then, the feedback-form strategy
$(\alpha^*_t)_{t\ge 0}$ defined by
\[
 \alpha^*_t 
:=a^*\left( X^*_t\right)
\quad t\ge 0
\]
is optimal for the problem \eqref{eq:MP} 
and $X^*\equiv X^{\alpha^*}$ holds.
Moreover, $U\equiv \bar{U}$ holds.
\end{thm}

\begin{proof}
See Appendix A.4. 
\end{proof}

\begin{rem}[Strong solvability of SDE \eqref{eq:SDE-X-global}]
\label{rem:wellposed_general}
In general, it seems difficult to establish the existence and 
pathwise uniqueness of a strong solution to the 
reflected SDE \eqref{eq:SDE-X-global} under the present level of generality.
The main reason is that the discontinuity of $x\mapsto a^*(x)$ may exist at the threshold $x=\ell$: more precisely, 
\[
a^*(\ell-)=\Pi_{\tilde K}(\beta(\ell-)q)
\quad\text{and}\quad
a^*(\ell+)=a^\star
\]
need not be aligned in the same direction as (i) $\beta(\ell-)\ne \beta(\ell+)$ in general, and (ii) the projection $\Pi_{\tilde{K}}(\cdot)$ depends on the geometry of the compact and convex set $\tilde K$.
This may destroy pathwise 
uniqueness for multidimensional SDEs with discontinuous diffusion coefficients.
Besides pursuing special cases in which {\it strong} well-posedness can
be proved, another possible approach would be to reformulate the
control problem in a {\it weak} formulation, 
where an admissible control is a system consisting of 
$(\Omega',{\mathcal F}',{\mathbb P}',({\mathcal F}'_t)_{t\ge 0}, W',\alpha)$. 
In such a framework, 
one would work with weak solutions (or, equivalently, the associated martingale problem) rather than insisting on a strong solution driven by a given Brownian motion. 
For the detailed explanation of the weak formulation of stochastic control problem, we refer, e.g., to Fleming and Soner~\cite[p.~135]{fleming_controlled_2006}, Nisio \cite[p.~35]{nisio_stochastic_2015}, and Yong and Zhou \cite[p.~64]{yong_stochastic_1999}, for example. 
\end{rem}

\begin{rem}[On discontinuous feedback strategies: relation to Brinker~\cite{brinker_minimal_2021}]
In portfolio optimization problems arising from mathematical finance,
one typically starts from a prescribed financial market model,
namely a fixed filtered probability space and a given driving Brownian motion.
From this viewpoint, 
the strong formulation is often more natural than the weak one.
Consequently, we think that it is important to understand 
whether the closed-loop reflected SDE admits a unique strong solution, 
even when the optimal feedback coefficient is discontinuous.
A related difficulty arises in Brinker~\cite{brinker_minimal_2021}, 
where the optimal feedback strategy also exhibits a discontinuity at a threshold.
In that setting, the author first establishes the existence of a weak solution
to the controlled SDE using general results for SDEs with discontinuous coefficients.
Next, pathwise uniqueness is obtained by exploiting the specific structure
of the model.
As a consequence, the existence of a unique strong solution follows from
the Yamada--Watanabe theorem.
This shows that discontinuities in the feedback control do not necessarily
preclude strong well-posedness, provided that additional structural properties
are available.
Compared with Brinker~\cite{brinker_minimal_2021}, our setting is more involved due to the geometry
of the set
$\tilde K=\{\sigma^\top x-\sigma_B \mid x\in K\}$
and the benchmark-relative formulation, which may lead to a genuinely
multidimensional and discontinuous diffusion coefficient.
This motivates the search in the next Section~4.4 
for geometric conditions under which
the closed-loop reflected SDE still admits a unique strong solution.
\end{rem}

\subsection{Geometric Structures Ensuring Strong Well-Posedness}

In this section, we investigate several geometric structures 
of the constraint set under which the closed-loop reflected SDE 
\eqref{eq:SDE-X-global} admits a unique strong solution,
despite the discontinuity of the optimal feedback coefficient.
The three examples below illustrate different mechanisms restoring strong
well-posedness.
\begin{itemize}
\item In Proposition~4.1, the projection remains continuous across the threshold,
so that the discontinuity disappears.

\item In Proposition~4.2, the projection is completely pinned to a single boundary point,
which yields a constant optimal feedback coefficient.

\item In Proposition~4.3, the multidimensional discontinuity is reduced to
an effectively one-dimensional structure, allowing the application of
the strong uniqueness theory for one-dimensional SDEs with discontinuous coefficients.
\end{itemize}
Consider a benchmark-relative portfolio constraint of the form
\begin{equation}
K=\{\pi_0 + c v \mid c\in[\underline c,\overline c]\} 
\ \subset {\mathbb R}^n, 
\label{eq:K-def}
\end{equation}
where $\pi_0\in\mathbb R^n$ is a reference portfolio, 
$v\in\mathbb R^n$ represents a fixed active trading direction, 
and $-\infty<\underline{c}<\overline{c}<\infty$.
Then
\[
\tilde K=\{a_0+cu\mid c\in[\underline c,\overline c]\}
\ \subset {\mathbb R}^d, 
\]
where 
\begin{equation}
a_0:=\sigma^\top \pi_0 - \sigma_B,
\quad u:=\sigma^\top v.
\label{eq:def-a0u} 
\end{equation}
The term $a_0$ represents a baseline mismatch between 
the attainable volatility $\sigma^\top \pi_0$
and the benchmark volatility $\sigma_B$, 
and the term $u$ describes the active exposure direction. 
In this case, the optimal feedback coefficient
is explicitely expressed as
\begin{equation}
a^*(x):=\Pi_{\tilde K}(\beta(x)q)
=a_0+\gamma(x)u
\label{eq:astar-explicit}
\end{equation}
with the scalar function 
$\gamma: {\mathbb R}_+ \to [\underline c,\overline c]$, given by
\begin{align*}
\gamma(x)
:=\Pi_{[\underline c,\overline c]}
\left(\frac{(\beta(x)q-a_0)^\top u}{|u|^2}\right) 
=\min\left\{ \overline{c}, 
\max\left\{ \underline{c}, \beta(x) m-\eta\right\}\right\},
\end{align*}
where $\Pi_{[\underline c,\overline c]}$
represents the metric projection on the closed interval
$[\underline c,\overline c]$ and
\begin{align}
m:=&\frac{q^\top u}{|u|^2}=
\frac{q^\top \sigma^\top v}{\left|\sigma^\top v\right|^2}, 
\label{eq:def-m} \\
\eta:=&\frac{a_0^\top u}{|u|^2}
=\frac{\left( \pi^\top_0\sigma-\sigma_B^\top\right)\sigma^\top v}
{\left| \sigma^\top v\right|^2}.
\label{eq:def-eta}
\end{align}
We see the following.
\begin{prop}[Continuous Projection Across the Threshold]
\label{prop:continuity_threshold}
Assume that 
\begin{equation}
p+q^\top a-\frac{|a|^2}{2}>0
\quad\text{with some $a\in \tilde{K}$.}
\label{eq:suff-drift}
\end{equation}
Then, the following are valid.
\begin{itemize}
 \item[\rm (1)] The function $U$ defined by \eqref{eq:def-U}
satisfies $U''(x):=\kappa^\star V_1''(x)>0$ for all $x\in (0,\ell)$, i.e., 
$U$ is strictly convex on $(0,\ell)$. 
Consequently, it holds that
\[
0<\beta(x)<1 \quad \text{for }{}^\forall x\in(0,\ell)
\quad\text{and}\quad
\beta(x)=\frac{1}{1-\bar r}
\quad \text{for }{}^\forall x\ge \ell.
\]
This identifies 
the effective range of the projection parameter function $x\mapsto \beta(x)$ given by~\eqref{eq:def-beta-U}.
 \item[\rm (2)] Suppose one of the following conditions holds:
\begin{align*}
{\rm (a)}&\quad
\beta(\ell-)m-\eta\ge \overline{c} \quad\text{and}\quad
\frac{m}{1-\bar{r}}-\eta \ge \overline{c}, \\
{\rm (b)}&\quad
\beta(\ell-)m-\eta\le \underline{c} \quad\text{and}\quad
\frac{m}{1-\bar{r}}-\eta\le \underline{c},
\end{align*}
where we use \eqref{eq:def-m} and \eqref{eq:def-eta}.
These conditions guarantee that the projection
\[
\Pi_{\tilde K}(\beta(x)q)
\]
remains on the same endpoint of the interval constraint 
across the threshold $x=\ell$.
More precisely,
\begin{itemize}
\item under condition {\rm (a)},
$a^*(\ell-)=a^*(\ell+)=a_0+\overline c\,u$,
\item under condition {\rm (b)},
$a^*(\ell-)=a^*(\ell+)=a_0+\underline c\,u$.
\end{itemize}
Hence, the optimal feedback coefficient $x\mapsto a^*(x)$ 
given in Theorem 4.4 (3) is continuous on
$[0,\infty)$.
In particular, the SDE 
\eqref{eq:SDE-X-global} admits a unique strong solution.
\end{itemize}
\end{prop}
\begin{proof}
See Appendix A.5.
\end{proof}
\begin{rem}
(1) The condition \eqref{eq:suff-drift} 
means that there exists an admissible benchmark-relative 
volatility exposure for which 
the relative log-growth is positive. 
See \eqref{eq:L}, \eqref{eq:def-pq}, and the following explanation.
\medskip

\noindent{(2)} 
Recalling \eqref{eq:def-m}, 
the sign of $m$ determines how the direction $u$ aligns with 
the benchmark-relative risk premium $q$.
\begin{itemize}
\item If $m>0$, the direction $u$ is positively aligned with $q$, meaning that taking risk in this direction yields a positive benchmark-relative risk premium. In this case, the investor is driven toward the upper bound $\overline c$.

\item If $m<0$, the direction $u$ is negatively aligned with $q$, and taking risk in this direction reduces benchmark-relative performance. The investor is therefore pushed toward the lower bound $\underline c$.

\item If $m=0$, the direction $u$ is orthogonal to $q$, and hence does not contribute to benchmark-relative growth. In this case, the feedback coefficient is constant and no discontinuity arises.
\end{itemize}
\end{rem}

The following case can be viewed as an extreme form 
of the boundary-binding behavior in Proposition 4.1, 
where the projection remains on the same face 
of the constraint set for all values of the state variable.
\begin{prop}[Complete Boundary Binding]
\label{prop:complete_binding}
Suppose that there exists a point $a^\dagger \in \tilde K$ such that
\begin{equation}
\label{eq:comp_bind_proj}
\Pi_{\tilde K}(tq)=a^\dagger
\quad \text{for all } t\in \left[0,\frac{1}{1-\bar r}\right].
\end{equation}
Then the optimal feedback coefficient is constant:
\[
a^*(x)\equiv a^\dagger,
\qquad x\ge 0.
\]
In particular, the closed-loop reflected SDE \eqref{eq:SDE-X-global} 
admits a unique strong solution.
\end{prop}
\begin{proof}
By Proposition 4.1 (1), 
Since $\beta(x)\in [0,1/(1-\bar r)]$, the assumption implies
$a^*(x)=\Pi_{\tilde K}(\beta(x)q)=a^\dagger$ for all $x\ge 0$.
\end{proof}

\begin{rem}[Geometric Interpretation]
\label{rem:complete_binding_geometry}
The condition in Proposition \ref{prop:complete_binding} admits a natural geometric interpretation.
Indeed, the relation 
\eqref{eq:comp_bind_proj}
means that the entire ray
$\{tq\mid t\ge 0\}$
remains inside the normal cone of $\tilde K$ at the point $a^\dagger$.
Equivalently, the projection of the target signal $\beta(x)q$
onto the admissible set $\tilde K$
always selects the same point $a^\dagger$.
Hence, Proposition \ref{prop:complete_binding}
can be viewed as an extreme form of the boundary-binding behavior
described in Proposition \ref{prop:continuity_threshold},
where the projection remains on the same face of the constraint set
for all values of the state variable.
\end{rem}

\begin{rem}[Portfolio Management Interpretation]
\label{rem:complete_binding_finance}
The complete binding regime in Proposition \ref{prop:complete_binding}
also admits a natural interpretation from the viewpoint of portfolio management.
Recall that
\[
\tilde K
=
\{\sigma^\top x-\sigma_B \mid x\in K\},
\]
where $K$ denotes the original portfolio constraint set.
If there exists $\pi^\dagger\in K$ such that
\[
a^\dagger
=
\sigma^\top \pi^\dagger-\sigma_B,
\]
then Proposition \ref{prop:complete_binding}
states that the optimal benchmark-relative volatility exposure
is constantly given by $a^\dagger$,
independently of the state variable.
This situation may arise when the admissible set $K$
is sufficiently restrictive,
for instance under strong position limits,
mandate constraints,
or concentrated allocation rules,
so that the optimal benchmark-relative portfolio
always remains on the same boundary point of the feasible region.
Typical examples include box constraints,
simplex-type long-only constraints,
and convex combinations of finitely many prescribed model portfolios.
\end{rem}
\noindent
Further, we obtain the following, 
which admits strong existence and uniqueness of the optimal 
state SDE~\eqref{eq:SDE-X-global}, 
even though the coefficient $x\mapsto a^*(x)$ may be discontinuous. 
\begin{prop}[Effective $1$-dimensional Projection Structure]
\label{prop:effective_1d_projection}
Assume that there exist a nonzero vector $u\in\mathbb R^d$,
a compact interval $I\subset \mathbb R$,
and a compact convex set 
$0\in C\subset u^\perp:=
\left\{ v\in {\mathbb R}^d \ \middle| \ u^\top v=0\right\}$
such that
\begin{equation}
\label{eq:effective_const}
\tilde{K} =\{cu+z \mid c\in I,\ z\in C\}. 
\end{equation}
Assume further that
\[
q\in \operatorname{span}\{u\}
\quad\text{and}\quad
0\notin I.
\]
Then, for every $t\ge 0$,
\begin{equation}
\label{eq:prop4.3proj}
\Pi_{\tilde K}(tq)
=\Pi_I\!\left(
t\frac{q^\top u}{|u|^2}
\right)u.
\end{equation}
Consequently, the optimal feedback coefficient admits the representation
\begin{equation}
\label{eq:prop4.3opt} 
a^*(x)=\gamma^*(x)\,u,
\quad x\ge 0,
\end{equation}
where
\begin{equation}
\label{eq:prop4.3gamma} 
\gamma^*(x)
=
\Pi_I\!\left(
\beta(x)\frac{q^\top u}{|u|^2}
\right).
\end{equation}
In particular, the diffusion coefficient of the closed-loop reflected
SDE \eqref{eq:SDE-X-global} 
remains aligned with the fixed direction $u$.
Hence, the reflected SDE \eqref{eq:SDE-X-global} 
admits a unique strong solution.
\end{prop}


\begin{proof}
See Appendix A.6.
\end{proof}

\begin{rem}[Geometric Interpretation]
\label{rem:effective_1d_geometry}
Proposition \ref{prop:effective_1d_projection}
shows that the effective dimensionality of the projection problem
may be strictly smaller than the dimension of the admissible set $\tilde K$.
Indeed, although $\tilde K$ itself may be genuinely high-dimensional,
the conditions, \eqref{eq:effective_const} and $q\in \operatorname{span}\{u\}$,
imply that the orthogonal component $z\in C$
does not affect the projection of the target signal $tq$.
More precisely, the orthogonal decomposition
\[
|tq-(cu+z)|^2
=
|tq-cu|^2+|z|^2
\]
shows that the projection problem reduces to a one-dimensional optimization
along the direction $u$.
Consequently, the projection image remains aligned with the same direction $u$,
even though the admissible set $\tilde K$ is high-dimensional.
\end{rem}

\begin{rem}[Portfolio Management Interpretation]
\label{rem:effective_1d_finance}
The structure in Proposition \ref{prop:effective_1d_projection}
admits a natural interpretation in terms of portfolio management.
Recall that
\[
\tilde K
=
\{\sigma^\top x-\sigma_B \mid x\in K\},
\]
where $K$ denotes the original portfolio constraint set.
The condition \eqref{eq:effective_const}
means that the admissible benchmark-relative volatility exposures
consist of
\begin{itemize}
\item a one-dimensional effective risk direction $u$, and
\item additional directions orthogonal to $u$ that do not contribute
to the benchmark-relative risk premium.
\end{itemize}
Since the benchmark-relative risk premium vector satisfies
$q\in \operatorname{span}\{u\}$,
the optimization depends only on the scalar exposure along the direction $u$.
The orthogonal component $C$ acts merely as a risk-neutral slack variable
from the viewpoint of benchmark-relative performance.
From a financial perspective, this corresponds to a market environment
in which only a single effective factor drives benchmark-relative growth,
even though the admissible portfolio set itself may be high-dimensional.
\end{rem}
\section{Numerical Experiments}

The numerical experiments involve two correlated risky assets (correlation $\varrho=0.54$). 
The benchmark is the equal-weighted portfolio of both assets, while the investor may trade only Asset~1 and the risk-free asset. 
The figure shows the benchmark-relative drawdown averaged over $5000$ simulated paths, comparing the optimal drawdown-duration strategy with a constant $100\%$-Asset~1 (i.e., buy-and-hold) strategy.
We use the following parameters: 
\begin{itemize}
  \item Risk-free rate $r = 0.01$
  \item Volatility matrix and drift vector
        \[
          \sigma =
          \begin{pmatrix} 0.18 & 0 \\ 0.09 & 0.14 \end{pmatrix},
          \qquad
          \mu = (0.09,\, 0.08)
        \]
  \item Benchmark: equal-weighted portfolio
  \item Discount rate $\delta = 0.03$
  \item Critical drawdown barrier
        $\ell = -\log(0.8) \approx 0.223$
        (a $20\%$ benchmark-relative drawdown)
  \item Weight constraint $\pi \in [0,\,2]$ for Asset~1
  \item Time horizon $24$ years, weekly rebalancing,
        $5000$ paths
\end{itemize}

Figure~\ref{fig:underwater_curves} plots the interquartile range, median and mean of underwater curves from the Monte Carlo simulation. 
Note that the median dominates the mean, which is due to the occurrence of prolonged and deep drawdowns. 
For the observed time period and random seed, the 25th percentile of paths does not breach the critical drawdown level. 

\begin{figure}[htbp!]
    \centering
    \includegraphics[width=1\linewidth]{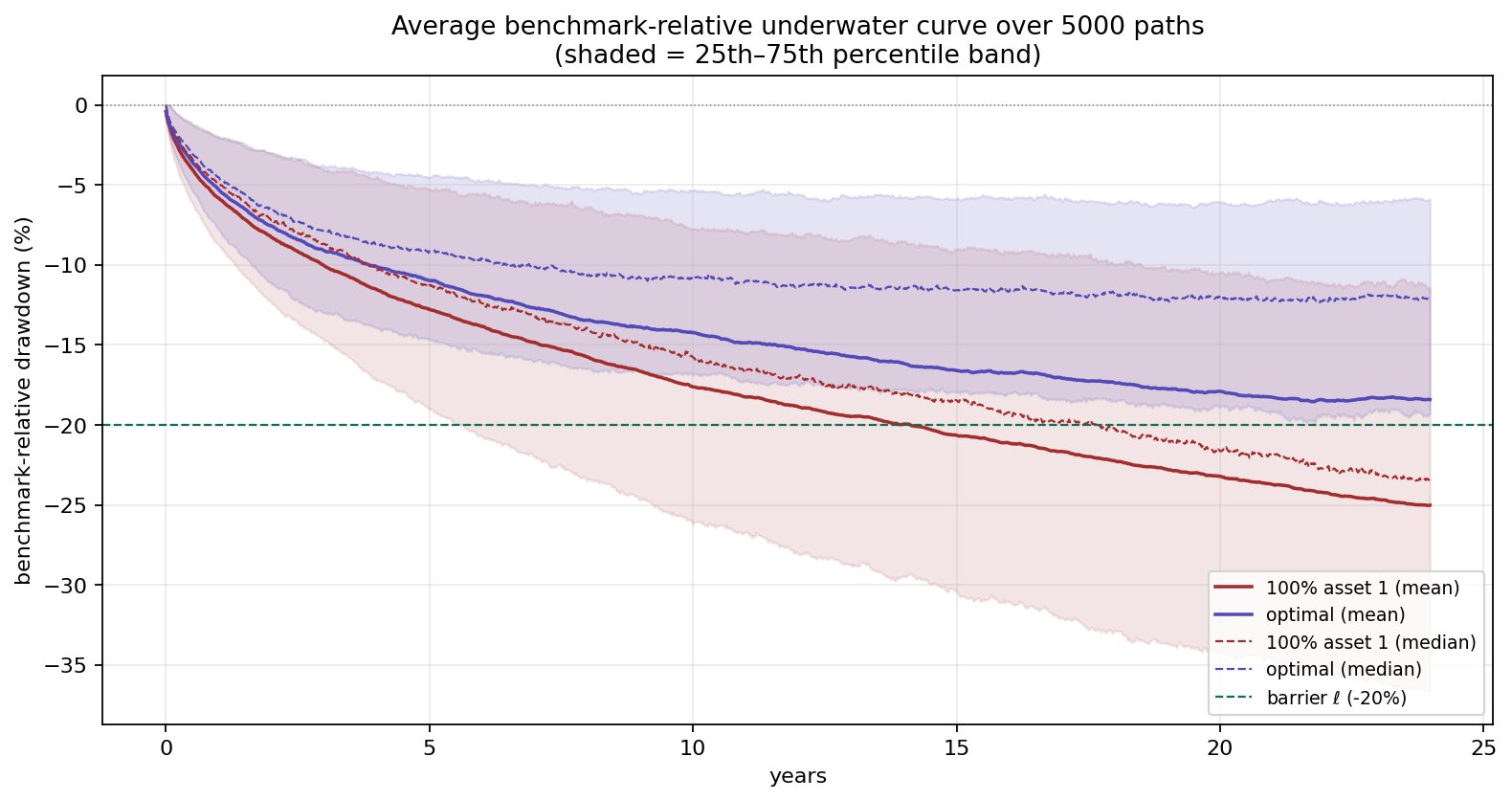}
    \caption{Underwater curves with interquartile ranges}
    \label{fig:underwater_curves}
\end{figure}

In Figure~\ref{fig:ecdf_crit_dd_time}, we show the empirical cumulative distribution functions for the time that the optimal strategy spends in critical drawdown. 
In the top-left panel ($d=10\%$), about 70\% of paths under the optimal strategy spend at most 40\% of the total time in critical drawdown, whereas for the buy-and-hold strategy, the corresponding proportion of paths is below 20\%. 

These numerical experiments illustrate the theoretical findings.
Indeed, rather than seeking to outperform a stochastic benchmark, the simulations indicate that the strategy mitigates sustained relative underperformance, as reflected in shorter durations of benchmark-relative critical drawdowns.

\begin{figure}[htbp!]
    \centering
    \includegraphics[width=1\linewidth]{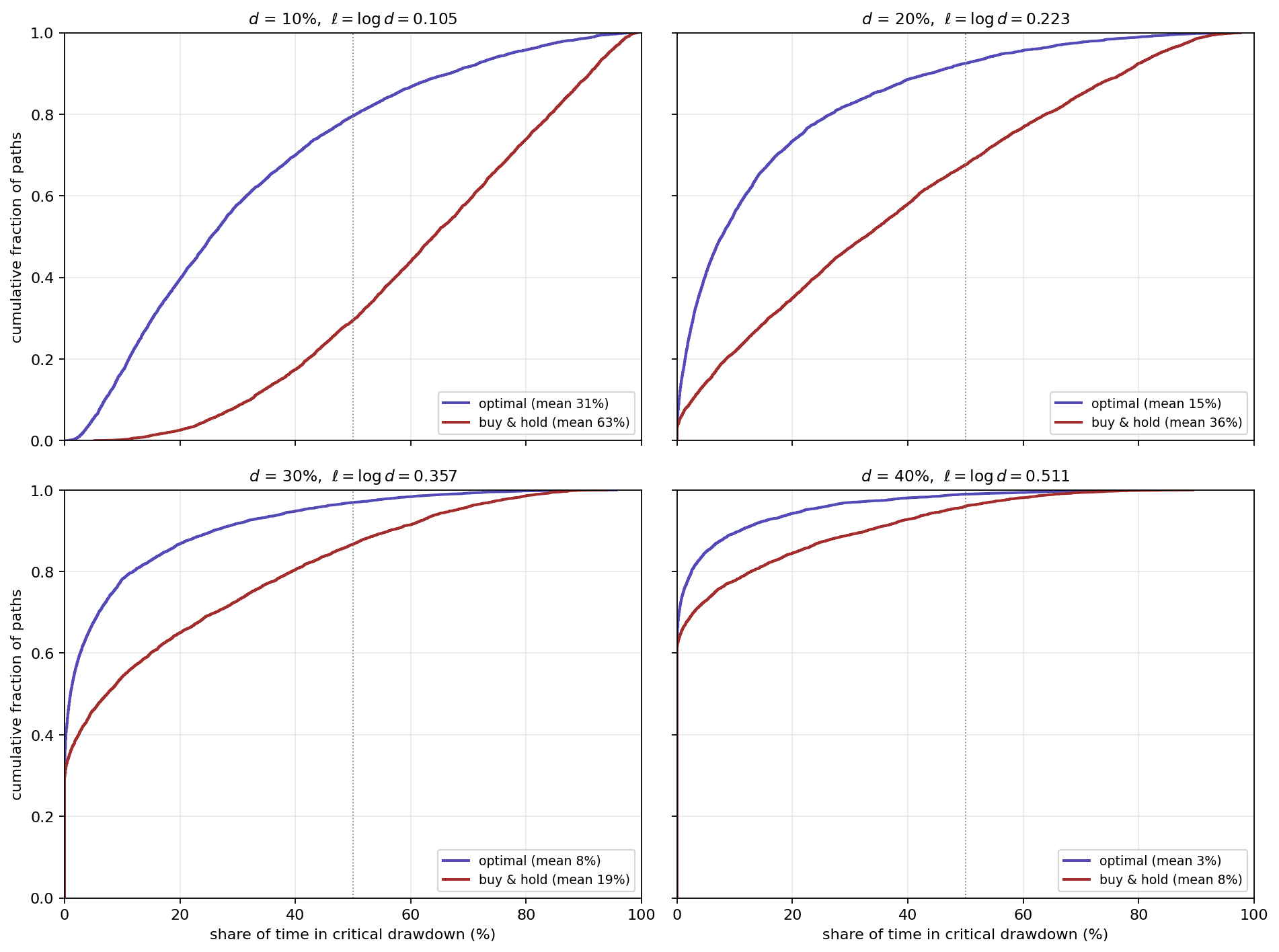}
    \caption{Empirical CDF of time in drawdown for different critical drawdown levels}
    \label{fig:ecdf_crit_dd_time}
\end{figure}


\appendix
\newpage
\section{Proofs}
\subsection{Proof of Theorem~\ref{thm:U-solves-HJB}}





\begin{proof}
Fix $\kappa>0$. 
Recall that $\tilde K\subset\mathbb R^d$ is nonempty, compact, and convex, and satisfies
\[
\mathfrak a := \inf_{a\in\tilde K}|a|>0,
\]
by Assumption~\ref{ass:non-replicability} (non-replicability of the benchmark). 
In particular, $|a|^2/2>0$ for all $a\in\tilde K$. 

\medskip
\noindent\textbf{Step 1: Reformulation of~\eqref{eq:HJB-Vk} as a first-order system.}
Let $w:=V_\kappa$ and $v:=V_\kappa'$. Define
\[
\mathcal F(w,v)
:=\sup_{a\in\tilde K}\frac{2}{|a|^2}\Bigl[\delta w+\Bigl(p+q^\top a-\frac{|a|^2}{2}\Bigr)v\Bigr].
\]
We claim that $V_\kappa$ solves~\eqref{eq:HJB-Vk} on $[0,\ell]$ if and only if $(w,v)$ satisfies
\begin{equation}\label{eq:system}
\begin{pmatrix}w\\v\end{pmatrix}'
=
G(w,v):=
\begin{pmatrix}v\\ \mathcal F(w,v)\end{pmatrix},
\qquad
\begin{pmatrix}w(0)\\v(0)\end{pmatrix}
=
\begin{pmatrix}\kappa\\0\end{pmatrix}.
\end{equation}
\\
(\eqref{eq:HJB-Vk} $\Rightarrow$ \eqref{eq:system}) 
Assume $V_\kappa\in C^2([0,\ell])$ solves~\eqref{eq:HJB-Vk}. 
For each fixed $x\in[0,\ell]$, denote $w=w(x)$, $v=v(x)$, and $v'=v'(x)$. 
Equation~\eqref{eq:HJB-Vk} reads
\[
-\delta w+\inf_{a\in\tilde K}\left\{\frac{|a|^2}{2}v'-
\Bigl(p+q^\top a-\frac{|a|^2}{2}\Bigr)v\right\}=0.
\]
Equivalently, for any $a\in\tilde K$,
\[
\frac{|a|^2}{2}v'\ge \delta w+\Bigl(p+q^\top a-\frac{|a|^2}{2}\Bigr)v,
\]
and, since $|a|>0$ for all $a\in\tilde K$, dividing by $|a|^2/2$ yields
\begin{align}\label{eq:vp_ge_F}
v' \ge \frac{2}{|a|^2}\Bigl[\delta w+\Bigl(p+q^\top a-\frac{|a|^2}{2}\Bigr)v\Bigr]
\qquad \forall a\in\tilde K.
\end{align}
Taking the supremum over $a\in\tilde K$ gives $v'\ge \mathcal F(w,v)$.
\\
On the other hand, since the map
\[
a\mapsto \frac{|a|^2}{2}v'-\Bigl(p+q^\top a-\frac{|a|^2}{2}\Bigr)v
\]
is continuous, there exists a minimizer $a^*(x)$ due to the compactness of the set $\tilde K$. 
At this $a^*(x)$, the infimum is attained and~\eqref{eq:HJB-Vk} implies
\[
\frac{|a^*(x)|^2}{2}v' = \delta w+\left(p+q^\top a^*(x)-\frac{|a^*(x)|^2}{2}\right)v
\]
so that 
\[
v' = \frac{2}{|a^*(x)|^2}\Bigl[\delta w+\Bigl(p+q^\top a^*(x)-\frac{|a^*(x)|^2}{2}\Bigr)v\Bigr]
\le \mathcal F(w,v).
\]
Together with~\eqref{eq:vp_ge_F}, this entails $v'=\mathcal F(w,v)$, so that, in light of $w'=v$ and the initial conditions, we obtain~\eqref{eq:system}.
\\
(\eqref{eq:system} $\Rightarrow$ \eqref{eq:HJB-Vk}) Conversely, assume $(w,v)$ satisfies \eqref{eq:system} on $[0,\ell]$ and set $V_\kappa:=w$. Then $V_\kappa\in C^2([0,\ell])$ with $V_\kappa'=v$ and $V_\kappa''=v'=\mathcal F(w,v)$.
By definition of $\mathcal F$, for all $a\in\tilde K$ we have
\begin{align*}
V_\kappa'' = \mathcal F(V_\kappa,V_\kappa') \ge \frac{2}{|a|^2}\Bigl[\delta V_\kappa +\Bigl(p+q^\top a-\frac{|a|^2}{2}\Bigr)V_\kappa'\Bigr],
\end{align*}
i.e.,
\begin{align}
\label{eq:a2vkappa}
\frac{|a|^2}{2}V_\kappa''-\left(p+q^\top a-\frac{|a|^2}{2}\right)V_\kappa' \ge \delta V_\kappa
\qquad \forall a\in\tilde K,
\end{align}
which continues to hold for the infimum in $\tilde K$.
Moreover, the supremum defining $\mathcal F(V_\kappa,V_\kappa')$ is attained due to continuity of $\mathcal F$ and compactness of $\tilde K$, so there exists $\hat a(x)\in\tilde K$ such that equality holds for $a=\hat a(x)$; hence the infimum in~\eqref{eq:a2vkappa} equals $\delta V_\kappa$ and \eqref{eq:HJB-Vk} holds with equality. 
Therefore, $V_\kappa$ is a classical solution of the HJB equation with $V_\kappa(0)=\kappa$ and $V_\kappa'(0)=0$.

\medskip
\noindent\textbf{Step 2: $G$ is globally Lipschitz.}
For each $a\in\tilde K$, define the linear map
\[
\phi_a(w,v):=\frac{2}{|a|^2}\Bigl[\delta w+\Bigl(p+q^\top a-\frac{|a|^2}{2}\Bigr)v\Bigr].
\]
Then $\mathcal F(w,v)=\sup_{a\in\tilde K}\phi_a(w,v)$.
Let $\bar a:=\sup_{a\in\tilde K}|a|<\infty$ (compactness). Since $|a|\ge \mathfrak a>0$ on $\tilde K$, the coefficients of $\phi_a$ are uniformly bounded:
\[
\left|\frac{2\delta}{|a|^2}\right|\le \frac{2\delta}{\mathfrak a^2},
\qquad
\left|\frac{2}{|a|^2}\Bigl(p+q^\top a-\frac{|a|^2}{2}\Bigr)\right|
\le \frac{2}{\mathfrak a^2}\Bigl(|p|+|q|\bar a+\frac{\bar a^2}{2}\Bigr).
\]
Hence there exists $L<\infty$ such that
\[
|\phi_a(w_1,v_1)-\phi_a(w_2,v_2)|\le L\bigl(|w_1-w_2|+|v_1-v_2|\bigr)
\quad \text{for all } a\in\tilde K.
\]
We thus obtain, for any $(w_1,v_1),(w_2,v_2)\in\mathbb R^2$,
\[
|\mathcal F(w_1,v_1)-\mathcal F(w_2,v_2)|
\le \sup_{a\in\tilde K}|\phi_a(w_1,v_1)-\phi_a(w_2,v_2)|
\le L\bigl(|w_1-w_2|+|v_1-v_2|\bigr), 
\]
so that $\mathcal F$ is globally Lipschitz, and therefore $G(w,v)=(v,\mathcal F(w,v))^\top$ is globally Lipschitz on $\mathbb R^2$.

\medskip
\noindent\textbf{Step 3: Existence and uniqueness (Picard-Lindel\"of).}
Since $G$ is globally Lipschitz, the Picard--Lindel\"of theorem yields a unique solution $(w,v)\in C^1([0,\ell];\mathbb R^2)$ of \eqref{eq:system}. Defining $V_\kappa:=w$, we get $V_\kappa\in C^2([0,\ell])$ with $V_\kappa'=v$ and $V_\kappa''=v'=\mathcal F(V_\kappa,V_\kappa')$, and the initial conditions hold by construction. 
By Step~1, this yields the unique classical solution $V_\kappa\in C^2([0,\ell])$ of the HJB equation~\eqref{eq:HJB-Vk}.

\medskip
\noindent\textbf{Step 4: Scaling $V_\kappa=\kappa V_1$.}
For any $\lambda>0$, positive homogeneity of the supremum implies
\[
\mathcal F(\lambda w,\lambda v)=\lambda \mathcal F(w,v),
\qquad\text{hence}\qquad
G(\lambda w,\lambda v)=\lambda G(w,v).
\]
Let $(w_1,v_1)$ be the unique solution of \eqref{eq:system} with $\kappa=1$. Then $(\kappa w_1,\kappa v_1)$ satisfies \eqref{eq:system} with initial condition $(\kappa,0)$, and by uniqueness it coincides with $(w,v)$. Therefore $V_\kappa=w=\kappa w_1=\kappa V_1$ on $[0,\ell]$.
\end{proof}
\subsection{Proof of Theorem~\ref{thm:normalized-V}}

\noindent(1) Note that~\eqref{eq:V-normalized} is well-defined by Lemma~\ref{lemma:positivity} (i). The other assertion is easy to see. 

\medskip
\noindent(2) 
The function $\beta:[0,\ell]\to {\mathbb R}_+$ is well-defined by Lemma~\ref{lemma:positivity} (ii), (iii) below.
Furthermore, 
\begin{align*}
&\frac{|a|^2}{2}V_\kappa''(x)
-\left(p+q^\top a-\frac{|a|^2}{2}\right)V_\kappa'(x) \\
&= \frac{V_\kappa''(x)+V_\kappa'(x)}{2}|a|^2-V_\kappa'(x)q^\top a 
-p V_\kappa'(x) \\
&\stackrel{\eqref{eq:def-beta-V}}{=}\frac{V_\kappa''(x)+V_\kappa'(x)}{2}
\left(|a-\beta(x)q|^2-|\beta(x)q|^2\right)
-p V_\kappa'(x).
\end{align*}
Owing to the positivity (iii) of Lemma~\ref{lemma:positivity},
the minimization problem admits a unique minimizer, and 
the expression $a^*(x)=\Pi_{\tilde K}(\beta(x)q)$ follows.

\medskip
\noindent(3) 
Note that the function $x\mapsto a^*(x)$ is Lipschitz continuous due to the Lipschitz continuity
of $x\mapsto \beta(x)$ (shown in Lemma B.2), and that of the projection map 
$\Pi_{\tilde{K}}: {\mathbb R}^n \to \tilde{K}$. 
Hence, the first assertion
follows from the standard result of SDE theory. 
Next, for verifying the optimality of the candidate strategy 
$\alpha^* \in {\mathscr A}_{\tilde{K}}$,
take $\alpha\in {\mathscr A}_{\tilde{K}}$ arbitrarily.
Recalling the dynamics \eqref{eq:dyn-X} of $X^\alpha$ 
and the fact that $\bar{L}^\alpha$ is flat off 
$\{ t\ge 0 | X^\alpha_t =0\}$, i.e., 
\begin{equation}
 \int_0^\infty 1_{\{ X_s>0\}} \d\bar{L}^\alpha_s =0, 
\label{eq:reflect}
\end{equation}
we use It\^o's formula to see that
\begin{align}
&e^{-\delta (t\wedge \tau_d)} V\left( X_{t\wedge \tau_d}\right)-V(X_0) 
\nonumber \\
=&\int_0^{t\wedge \tau_d} 
e^{-\delta s} V'\left( X_{s}\right) \d\bar{L}_s 
-\int_0^{t\wedge \tau_d} 
e^{-\delta s} V'\left( X_{s}\right) \alpha_s^\top \d W_s
\nonumber \\
+&\int_0^{t\wedge \tau_d} 
e^{-\delta s} 
\left[
-\delta V
+\frac{|\alpha_s|^2}{2}V''
-\left( p + q^\top \alpha_s -\frac{1}{2}|\alpha_s|^2\right)
V'\right]\left( X_{s}\right)\d s.
\end{align}
Here, the first term of the right-hand-side is $0$ by 
\eqref{eq:reflect}
and $V'(0+)=0$, 
the second term of the right-hand-side is a martingale, 
and the integrand of the third term of the right-hand-side 
is nonnegative as $V$ solves HJB equation \eqref{eq:HJB-V}. Hence, 
taking expectation of (5.2), we deduce
\[
 {\mathbb E}\left[ e^{-\delta (t\wedge \tau_d)}
V\left( X_{t\wedge \tau_d}\right) 
\middle| X_0=x\right] \ge V(x).
\]
Letting $t\to\infty$ and noting that $V(d)=1$, we obtain 
\[
 {\mathbb E}\left[ e^{-\delta \tau_d} \middle| X_0=x\right] \ge V(x).
\]
Taking supremum with respect to $\alpha\in {\mathscr A}_{\tilde{K}}$, 
we deduce that $\bar{V}(x)\ge V(x)$.
Next, setting $\alpha=\alpha^* \in {\mathscr A}_{\tilde{K}}$ in (4.4), 
we deduce
\[
e^{-\delta (t\wedge \tau^*_d)} V\left( X^*_{t\wedge \tau^*_d}\right)-V(X^*_0)  
=-\int_0^{t\wedge \tau^*_d} 
e^{-\delta s} V'\left( X^*_{s}\right) 
\left( \alpha^*_s \right)^\top \d W_s,
\]
where we recall $X^*= X^{\alpha^*}$ and write
$\tau^*_d:=\inf\left\{ t\ge 0 \middle| X^*_t>\log d\right\}$.
Taking expectation and letting $t\to\infty$, we obtain 
\[
 {\mathbb E}\left[ e^{-\delta \tau^*_d}
\middle| X^*_0=x\right] = V(x),
\]
from which we deduce $\bar{V}(x)=V(x)$ and the optimality 
of $\alpha^*\in {\mathscr A}_{\tilde{K}}$. 
\subsection{Proof of Theorem~\ref{thm:hjb-sol-W}}

\noindent{(1)} By definition, 
$r(a)$ is the unique positive root of $f_a(r)=0$. Since
\[
f_a(0)=-\delta<0,
\quad
\lim_{r\to\infty} f_a(r)\to\infty,
\]
it follows that
\[
f_a(r)< 0
\quad \text{for }{}^\forall r\in[0,r(a))
\quad\text{and}\quad
f_a\left( r(a)\right)=0.
\]
Since $\bar r\le r(a)$ for every $a\in\tilde K$, we obtain
\begin{equation}
f_a(\bar{r})\le 0
\quad \text{for }{}^\forall a\in\tilde K
\quad \text{and}\quad f_{a^\star}(\bar{r})=0.
\label{eq:a-star-max}
\end{equation}
We see the function \eqref{eq:W-explicit} satisfies, 
for $a\in\tilde K$, 
\[
-\delta W(x)+\frac{|a|^2}{2}W''(x)
-\left(p+q^\top a-\frac{|a|^2}{2}\right)W'(x)
=f_a\left(\bar{r}\right)W(x)\le 0.
\]
Taking the supremum over $a\in\tilde K$, 
and using equality for $a=a^\star$, we obtain
\begin{equation}
-\delta W(x)
+\sup_{a\in\tilde K}
\left[
\frac{|a|^2}{2}W''(x)
-\left(p+q^\top a-\frac{|a|^2}{2}\right)W'(x)
\right]
=f_{a^\star}\left(\bar{r}\right)W(x)=0.
\label{eq:HJB-astar}
\end{equation}
To ensure the boundary conditions, 
$W(\ell)=1$ and $\lim_{x\to\infty}W(x)=0$,
are immediate.

\medskip
\noindent{(2)}
We show a sketch of the proof as it is standard argument: 
Take $\alpha\in\mathscr A_K$ arbitrarily.
Applying Ito's formula to
$e^{-\delta t}W(X_t^\alpha)$ 
yields
\begin{multline*}
\d\bigl(e^{-\delta t}W(X_t^\alpha)\bigr)
= e^{-\delta t}W'(X^\alpha_t)\d W_t \\
+e^{-\delta t}
\left\{
-\delta W
+\frac{|\alpha_t|^2}{2}W''
-\left(p+q^\top \alpha_t-\frac12|\alpha_t|^2\right)W'
\right\}(X_t^\alpha)\,\d t
\end{multline*}
Noting that the second term of the right-hand-side of the above 
is nonpositive,
we take expectations and stopping at $\rho_d$  to see that
\[
W(x)\ge \mathbb{E}\left[e^{-\delta\rho_d} W\left(X_{\rho_d}^\alpha\right)\right]
=\mathbb{E}\left[e^{-\delta\rho_d}\right].
\]
Hence $W(x)\ge \bar{W}(x)$. 
Next, take the constant control 
$\alpha^*\equiv a^\star\in {\mathscr A}_{\tilde{K}}$, 
and write 
$X^*:=X^{\alpha^*}$ and 
$\rho^*_d:=\inf\left\{ t\ge 0 \middle| X^*_t\le \ell \right\}$.
We deduce that 
\[
W(x)= \mathbb{E}\left[e^{-\delta\rho^*_d} W\left(X^*_{\rho^*_d}\right)\right]
=\mathbb{E}\left[e^{-\delta\rho^\star_d}\right]
\]
as $a^\star$ is the maximizer in the maximization 
of \eqref{eq:HJB-astar}. 
Hence, the optimality of $\alpha^* \in {\mathscr A}_{\tilde{K}}$
and the equality $W\equiv \bar{W}$ follow.
\subsection{Proof of Theorem~\ref{thm:hjb-sol-U}}

\noindent{(1)} 
For $\kappa,C>0$, consider the function
\[
\bar{U}_\kappa(x):=
\begin{cases}
V_\kappa(x) & \text{for $0\le x\le \ell$},\\
-C e^{-\bar r(x-\ell)}+\dfrac{1}{\delta} & \text{for $x>\ell$},
\end{cases}
\]
where $V_\kappa$ in the right-hand-side 
is the unique solution on $[0,\ell]$ obtained in Theorem 4.1
and $\bar{r}$ is given by \eqref{eq:bar-r}.
Note that 
$\bar{U}_\kappa|_{[0,\ell]}=V_\kappa=\kappa V_1$ solves \eqref{eq:HJB-U-left}
with the boundary condition $\bar{U}'_\kappa(0+)=0$,
which has been shown in Theorem 4.1. 
Also, note that
$\bar{U}_\kappa|_{(\ell,\infty)}$ solves \eqref{eq:HJB-U-right}
as we deduce, for $x>\ell$, 
\begin{align}
&-\delta \bar{U}_\kappa(x)
+1
+\inf_{a\in\tilde K}
\left[
\frac{|a|^2}{2} \bar{U}_\kappa''(x)
-\left(p+q^\top a-\frac{|a|^2}{2}\right)
\bar{U}_\kappa'(x)
\right] \nonumber \\
=&-C e^{-\bar r(x-\ell)} \sup_{a\in\tilde{K}} f_a\left( \bar{r}\right) 
\nonumber \\
=&-C e^{-\bar r(x-\ell)} f_{a^\star}\left( \bar{r}\right)=0,
\label{eq:a-star-maximizer}
\end{align}
where we use \eqref{eq:a-star-max}.
In addition, we see 
$\lim_{x\to\infty} \bar{U}_\kappa(x)={1}/{\delta}$.
For satifying 
other boundary conditions in \eqref{eq:HJB-U-b}:
\begin{equation}
\label{eq:matching}
\begin{split}
\bar{U}_\kappa(\ell-)=&V_{\kappa}(\ell-)=\kappa V_1(\ell-)
=-C+\frac{1}{\delta}=\bar{U}_{\kappa}(\ell+), \\
\bar{U}'_\kappa(\ell-)=&{V}_\kappa'(\ell-)=\kappa V_1'(\ell-)
=\bar r C=\bar{U}'_\kappa(\ell+),
\end{split}
\end{equation}
we solve the simultaneous linear equations \eqref{eq:matching} 
with respect to $(\kappa,C)$. The solution 
$\left( \kappa^\star, C^\star\right)$ is given by
\eqref{eq:kappa-star-C-star}, 
where we recall
\[
V_1(\ell)+\frac{V_1'(\ell-)}{\bar r}>0
\]
as $V_1(\ell)>0$, $V_1'(\ell-)>0$ and $\bar{r}>0$.
Hence, we obtain the function \eqref{eq:def-U} that solves
HJB equation \eqref{eq:HJB-U-left}-\eqref{eq:HJB-U-b}.

\medskip 
\noindent{(2)}
For $x \in (0,\ell)\cup (\ell,\infty)$, we see
\begin{align}
&\frac{|a|^2}{2}U''(x)
-\left(p+q^\top a-\frac{|a|^2}{2}\right)U'(x) 
\nonumber \\
=&\frac{U''(x)+U'(x)}{2}|a|^2-U'(x)q^\top a -p U'(x) 
\nonumber \\
=&
\frac{U''(x)+U'(x)}{2}
\left(|a-\beta(x)q|^2-|\beta(x)q|^2\right)-p U'(x).
\label{eq:comp-sq}
\end{align}
Here, note that $U''(x)+U'(x)>0$ holds 
for $0<x<\ell$ by (iii) of Lemma~\ref{lemma:positivity}, 
and
\[
 U''(x)+ U'(x)=C^\star\bar{r}\left( 1-\bar{r}\right)>0
\quad \text{and}\quad
\beta(x)=\frac{1}{1-\bar{r}}>0
\]
for $x>\ell$ as $C^\star>0$ and $0<\bar{r}<1$
(i.e., Assumption 4.1). 
Therefore, the expression of the minimizer, 
$a^*(x)=\Pi_{\tilde K}(\beta(x)q)$ follows for any
$x \in (0,\ell)\cup (\ell,\infty)$ from \eqref{eq:comp-sq}.
Further, we see $a^*(x)=a^\star$ for $x>\ell$
by \eqref{eq:a-star-max} and \eqref{eq:a-star-maximizer}.

\medskip 
\noindent{(3)}
Take $\alpha\in\mathscr A_{\tilde{K}}$ arbitrarily.
Since $U\in C^1([0,\infty))\cap C^2((0,\ell)\cup (\ell,\infty))$,
we apply the generalized It\^o formula to see 
\begin{multline}
U\left(X_t^\alpha\right)
=U(x)
+\int_0^t U'\left(X_s^\alpha \right)\d X_s^\alpha \\
+\frac12\int_0^t U''
\left(X_s^\alpha \right)\mathbf{1}_{\{X_s^\alpha \neq \ell\}}\,
\d\langle X^\alpha \rangle_s 
+\frac12\left\{ U'(\ell+)-U'(\ell-)\right\}L_t^\ell\left(X^\alpha \right),
\label{eq:gen-Ito}
\end{multline}
where 
$L_t^\ell(X^\alpha)$ denotes the local time 
of $X^\alpha$ at the level $\ell$.
By the smooth-fit condition in 
\eqref{eq:HJB-U-b}, the local-time term vanishes.
Moreover, since $\d\bar L^\alpha_t$ acts only when $X_t^\alpha=0$, 
and $U'(0+)=0$, we have
\[
\int_0^t U'\left(X_s^\alpha \right)\,\d\bar L^\alpha_s = 0.
\]
Hence, recalling $\d\langle X^\alpha \rangle_t = |\alpha_t|^2\,\d t$, 
\eqref{eq:gen-Ito} reduces to
\begin{multline*}
U\left(X_t^\alpha\right)
=U(x)
-\int_0^t U'\left(X_s^\alpha\right)\alpha_s^\top\,\d W_s \\
+\int_0^t
\left[
\frac{|\alpha_s|^2}{2}U''\left(X_s^\alpha \right)
-\left(p+q^\top \alpha_s-\frac{|\alpha_s|^2}{2}\right)U'
\left(X_s^\alpha \right)\right]\d s.
\end{multline*}
So, we deduce 
\begin{multline}
\d\left(e^{-\delta t}U\left(X_t^\alpha\right)\right)
=-e^{-\delta t}U'\left(X_t^\alpha \right)\alpha_t^\top\,\d W_t  \\
+e^{-\delta t}
\left[
-\delta U(X_t^\alpha)
+\frac{|a_t|^2}{2}U''\left(X_t^\alpha \right)
-\left(p+q^\top \alpha_t-\frac{|\alpha_t|^2}{2}\right)
U'(X_t^\alpha)
\right]\d t.
\label{eq:discounted-Ito}
\end{multline}
Since $U$ solves
\eqref{eq:HJB-U-left} and \eqref{eq:HJB-U-right}, we have
\[
-\delta U(x)
+\frac{|a|^2}{2}U''(x)
-\left(p+q^\top a-\frac{|a|^2}{2}\right)U'(x)
\ge -\mathbf{1}_{\{x>\ell\}},
\quad
\text{for all $a \in \tilde{K}$ and $x\neq \ell$}.
\]
Let 
$\tau_n:=\inf\{t\ge 0: X_t^\alpha \ge n\}$
for $n\in\mathbb N$ and $T\in {\mathbb R}_{++}$.
Integrating up to $T\wedge \tau_n$ 
and taking expectations, we obtain
%
\begin{align}
U(x)
&\le
{\mathbb E}\left[
\int_0^{T\wedge\tau_n} e^{-\delta s}
\mathbf{1}_{\left\{X_s^\alpha>\ell\right\}}\,\d s
+e^{-\delta(T\wedge\tau_n)}U(X_{T\wedge\tau_n}^\alpha)
\middle| X_0^\alpha=x\right],
\label{eq:verification-ineq}
\end{align}
Letting $n\to\infty$, the dominated convergence theorem yields
\[
U(x)
\le
{\mathbb E}\left[
\int_0^T e^{-\delta s}\mathbf{1}_{\left\{X_s^\alpha>\ell\right\}}\,\d s
+e^{-\delta T}U(X_T^\alpha)
\middle| X_0^\alpha=x\right].
\]
Now, recalling 
$\lim_{T\to\infty}{\mathbb E}\left[e^{-\delta T}u(X_T^\alpha)
\middle| X_0^\alpha=x\right]= 0$
as $U$ is bounded, 
we deduce 
\[
U(x)
\le
{\mathbb E}\left[
\int_0^\infty e^{-\delta s}\mathbf{1}_{\{X_s^\alpha>\ell\}}\,\d s
\middle| X_0^\alpha=x \right].
\]
Since $\alpha\in\mathscr A_K$ was arbitrary, we conclude that
$U(x)\le \bar{U}(x)$.
Next, note that the minimizer function $a^*$ satisfies
\[
-\delta U(x)
+\frac{|a^*(x)|^2}{2}U''(x)
-\left(p+q^\top a^*(x)-\frac{|a^*(x)|^2}{2}\right)U'(x)
=-\mathbf{1}_{\{x>\ell\}},
\quad x\neq \ell.
\]
Repeating the above argument with $\alpha^*:=(\alpha^*_t)_{t\ge 0}$,
$\alpha^*_t:=a^*(X^*_t)$, the inequality
\eqref{eq:verification-ineq} becomes an equality, and we obtain
\[
U(x)
=
{\mathbb E}\left[
\int_0^\infty e^{-\delta s}\mathbf{1}_{\{X_s^{\alpha^*}>\ell\}}\,\d s
\middle| X_0^{\alpha^*}=x\right].
\]
Hence, we conculde $U(x)=\bar{U}(x)$ ($x\ge 0$)
and the optimality of $\alpha^*\in {\mathscr A}_{\tilde{K}}$.
\subsection{Proof of Proposition 4.1}

\noindent{(1)} 
From \eqref{eq:HJB-U-left}, we see that, 
for all $a\in \tilde K$ and $0<x<\ell$, 
\begin{align*}
U''(x) \ge 
\frac{2}{|a|^2}
\left[
\delta U(x)+\left(p+q^\top a-\frac{|a|^2}{2}\right)U'(x) \right].
\end{align*}
Hence, $U''(x)>0$ for all $x\in (0,\ell)$ follows
by \eqref{eq:suff-drift} and (i)-(ii) of Lemma~\ref{lemma:positivity}
in Appendix~\ref{sec:positivities}, recalling $U|_{[0,\ell)}=\kappa V_1|_{[0,\ell)}$.
Therefore, $0<\beta(x)<1$ 
for all $x\in (0,\ell)$ also follows, using (ii) of Lemma B.1 again.

\medskip

\noindent{(2)}
Under Condition (a), we deduce
\begin{align*}
\gamma(\ell-)=&\min\left\{ \overline{c},
\max\left\{ \underline{c}, \beta(\ell-)m-\eta\right\}\right\}
 =\overline{c}, \\
\gamma(\ell+)=&\min\left\{ \overline{c},
\max\left\{ \underline{c}, \frac{m}{1-\bar{r}}-\eta\right\}\right\}
 =\overline{c}, 
\end{align*}
hence, the continuity of $x\mapsto \gamma(x)$ at $x=\ell$ follows. 
Case with Condition (b) is seen similarly.
For SDE \eqref{eq:SDE-X-global} with the coefficient $x\mapsto a^*(x)$, 
which is bounded and Lipschitz continuous,  
the existence and uniqueness of the strong solution 
follows from a standard result of reflecting SDEs.

\subsection{Proof of Proposition \ref{prop:effective_1d_projection}}

\begin{proof}[Proof of Proposition \ref{prop:effective_1d_projection}]
We first show that the projection problem is effectively one-dimensional.
By assumption,
\[
\tilde K
= \{cu+z\mid c\in I,\ z\in C\},
\quad 0\in C\subset u^\perp,
\quad\text{and}\quad q\in \operatorname{span}\{u\}.
\]
Hence we see that
\[
q=\tilde q\,u \quad\text{with}\quad
\tilde q:=\frac{q^\top u}{|u|^2}\in\mathbb R.
\]
Let \(t\ge 0\). For \(a=cu+z\in\tilde K\), where \(c\in I\) and \(z\in C\subset u^\perp\), we have
\[
|tq-a|^2
=
|t\tilde q\,u-(cu+z)|^2
=
|(t\tilde q-c)u-z|^2.
\]
Since \(z\perp u\), this becomes
\[
|tq-a|^2
=
|u|^2|t\tilde q-c|^2+|z|^2.
\]
Because \(0\in C\), the minimum over \(z\in C\) is attained at \(z=0\). 
Therefore, the relation \eqref{eq:prop4.3proj} and 
the expression \eqref{eq:prop4.3opt}
with \eqref{eq:prop4.3gamma} follows. 
Since \(0\notin I\) and \(I\) is compact, 
there exists \(\varepsilon_0>0\) such that
\[
\left|\gamma^*(x)\right|\ge \varepsilon_0
\qquad\text{for all }x\ge 0.
\]
Thus the diffusion coefficient is uniformly nondegenerate and remains aligned with the fixed direction $u$.
We now rewrite the closed-loop SDE \eqref{eq:SDE-X-global} as
\begin{align}
\d{X}^*_t&= -\gamma^*\left( X^*_t\right)u^\top \d W_t 
-\left\{ p+ \gamma^* \left( X^*_t\right) q^\top u
-\frac{1}{2}
\left| \gamma^*\left( X^*_t\right)u\right|^2
\right\}\d t+ \d \overline{L^*_t} \nonumber \\
&=-\tilde{\gamma}^*(X_t) \,\d W^0_t -b^*(X_t)\,\d t 
+ \d \overline{L^*_t},
\label{eq:SDE-1dim-Xstar}
\end{align}
where we define the one-dimensional Brownian motion
\[
W^0_t:=\frac{u^\top W_t}{|u|}
\quad (t\ge 0), 
\]
and 
\[
\tilde{\gamma}^*(x):=|u|\gamma^*(x),
\quad 
b(x):= p+ \gamma^*(x) q^\top u -\frac12|\gamma^*(x)u|^2.
\]
Recall that the coefficients \(\tilde{\gamma}^*\) and \(b^*\) 
are bounded and 
Lipschitz continuous on \((0,\ell)\) and constant on \([\ell,\infty)\). 
Thus the only possible discontinuity occurs at $x=\ell$.
Moreover, $\tilde{\gamma}^*$ is uniformly nondegenerate:
\[
|\tilde{\gamma}^*(x)|\ge |u|\varepsilon_0>0,
\quad (x\ge 0).
\]
By the standard weak existence result for one-dimensional reflected SDEs 
with bounded measurable drift and bounded uniformly nondegenerate 
diffusion coefficient, equation  
\eqref{eq:SDE-1dim-Xstar}
admits a weak solution.
It remains to prove pathwise uniqueness.
Observe that the diffusion coefficient $\tilde{\gamma}^*$ is scalar-valued, bounded, uniformly nondegenerate, and piecewise Lipschitz continuous with at most one discontinuity point. 
Moreover, the drift coefficient $b^*$ is bounded and piecewise Lipschitz continuous.
Hence the reflected SDE~\eqref{eq:SDE-1dim-Xstar} belongs to the class of one-dimensional reflected SDEs with discontinuous coefficients treated by Le Gall~\cite{le_gall_one-dimensional_1984} and Bass--Chen \cite{bass_one-dimensional_2005}.
Therefore, the pathwise uniqueness of the solution of \eqref{eq:SDE-1dim-Xstar} follows.  
Also, the unique existence of the strong solution of 
\eqref{eq:SDE-1dim-Xstar} (and \eqref{eq:SDE-X-global})
follows by the Yamada--Watanabe theorem.
\end{proof}

\section{Positivities and a Lipschitz Continuity}\label{sec:positivities}

In this sections, we prepare two lemmata, both of which are useful for showing our main theorems.
\begin{lem}\label{lemma:positivity}
For the solution $V_\kappa\in C^2([0,\ell])$ ($\kappa>0$) of 
\eqref{eq:HJB-Vk} in Theorem 4.1, the following properties hold:
\begin{itemize}
\item[\rm (i)] $V_\kappa(x)>0$ for all $x\in[0,\ell]$;
\item[\rm (ii)] $V_\kappa'(x)>0$ for all $x\in(0,\ell]$; 
\item[\rm (iii)] $V_\kappa''(x)+V_\kappa'(x)>0$ for all $x\in[0,\ell]$. 
\end{itemize}
\end{lem}

\begin{proof}
Since $V_\kappa(0)=\kappa>0$ and $V_\kappa'(0)=0$, we first evaluate 
\begin{equation}
V_\kappa''(x)
=
\sup_{a\in\tilde K}
\frac{2}{|a|^2}
\left\{
\delta V_\kappa(x)
+
\left(
p+q^\top a-\frac{|a|^2}{2}
\right)V_\kappa'(x)
\right\}.
\label{eq:Upp-sup}
\end{equation}
at $x=0$ and obtain
\[
V_\kappa''(0)
=
\sup_{a\in\tilde K}
\frac{2}{|a|^2}
\left\{
\delta V_\kappa(0)
+\left(p+q^\top a-\frac{|a|^2}{2}\right)V_\kappa'(0)
\right\}.
\]
Because $V_\kappa'(0)=0$, this simplifies to
\[
V_\kappa''(0)=\sup_{a\in\tilde K}\frac{2\delta\kappa}{|a|^2}>0
\]
by Assumption 4.1. 
By continuity of $V_\kappa''$, there exists $\varepsilon>0$ such that
\[
V_\kappa''(x)>0
\qquad \text{for all } x\in[0,\varepsilon].
\]
Hence $V_\kappa'$ is strictly increasing on $[0,\varepsilon]$, 
and since $V_\kappa'(0)=0$, we obtain
\[
V_\kappa'(x)>0
\qquad \text{for all } x\in(0,\varepsilon].
\]
Consequently,
\[
V_\kappa(x)
=
\kappa+\int_0^x V_\kappa'(y)\,\d y
>
\kappa
>
0
\qquad \text{for all } x\in(0,\varepsilon].
\]
We next prove that $V_\kappa'(x)>0$ for all $x\in(0,\ell]$. 
Suppose, for contradiction,
that there exists $x_1\in(0,\ell]$ such that 
$V_\kappa'(x_1)\le0$. Since $V_\kappa'>0$ on
$(0,\varepsilon]$, we may define
\[
x_0:=\inf\{x>\varepsilon:\,V_\kappa'(x)\le0\}\in[\varepsilon,\ell].
\]
By continuity of $V_\kappa'$, we have
\[
V_\kappa'(x_0)=0,
\quad
V_\kappa'(x)>0 \quad \text{for } x\in(0,x_0).
\]
Therefore
\[
V_\kappa(x_0)
=
\kappa+\int_0^{x_0}V_\kappa'(y)\,\d y >\kappa> 0.
\]
Evaluating \eqref{eq:Upp-sup} at $x=x_0$, we obtain
\[
V_\kappa''(x_0)
=
\sup_{a\in\tilde K}
\frac{2}{|a|^2}
\left\{ \delta V_\kappa(x_0)
+\left(p+q^\top a-\frac{|a|^2}{2}\right)V_\kappa'(x_0) \right\}
=\sup_{a\in\tilde K}\frac{2\delta V_\kappa(x_0)}{|a|^2}
>0.
\]
Since $V_\kappa'$ decreases from strictly positive values to 0 at $x_0$, $V''_\kappa(x_0) \le 0$, contradicting the strict positivity of the curvature of $V_\kappa$ established above. 
Hence
\[
V_\kappa'(x)>0
\qquad \text{for all } x\in(0,\ell].
\]
It follows immediately that
\[
V_\kappa(x)=\kappa+\int_0^xV_\kappa'(y)\,\d y>\kappa>0
\qquad \text{for all } x\in[0,\ell].
\]
Finally, adding $V_\kappa'(x)$ to both sides of \eqref{eq:Upp-sup}, we obtain
\[
V_\kappa''(x)+V_\kappa'(x)
=
\sup_{a\in\tilde K}
\frac{2}{|a|^2}
\left\{ \delta V_\kappa(x)
+ \left(p+q^\top a\right)V_\kappa'(x)
\right\}.
\]
Since $V_\kappa(x)>0$ and $V_\kappa'(x)>0$, 
the quantity inside the braces is positive for
sufficiently large $q^\top a$, and therefore the supremum is strictly positive.
Hence
\[
V_\kappa''(x)+V_\kappa'(x)>0
\qquad \text{for all } x\in[0,\ell].
\]
This completes the proof.
\end{proof}
\noindent
Further, we obtain the following.
\begin{lem}
The function $\beta: [0,\ell]\to {\mathbb R}_{+}$
given by \eqref{eq:def-beta-V} is Lipschitz continuous.
\end{lem}
\begin{proof}
Set the two continuous functions on the 
compact interval $[0,\ell]$ as
\[
N(x):=V_\kappa'(x), \quad D(x):=V_\kappa''(x)+V_\kappa'(x).
\]
By Lemma~4.4, 
$D(x)>0$ for all $x\in[0,\ell]$. Further, note that 
\[
m:=\min_{x\in [0,\ell]}>0.
\]
Since $V_\kappa\in C^2([0,\ell])$, the function $N=V_\kappa'$ is Lipschitz on
$[0,\ell]$. It remains to prove that $V_\kappa''$ is Lipschitz: 
Write 
\[
\displaystyle V_\kappa''(x)=\sup_{a\in\tilde K} F_a(x)
\quad (x\in[0,\ell]), 
\]
where we define
\[
F_a(x):=\frac{2}{|a|^2}
\left\{
\delta V_\kappa(x)
+\left(p+q^\top a-\frac{|a|^2}{2}\right)V_\kappa'(x)
\right\}.
\]
for $a\in\tilde K$.
We deduce that 
\[
\displaystyle L:=\sup_{(x,a)\in [0,\ell]\times \tilde{K}}
\left| F_a'(x)\right|<\infty
\]
since $\tilde K$ is compact, $\underline{a}:=\inf_{a\in\tilde K}|a|>0$, and 
$V_\kappa',V_\kappa''$ are continuous on $[0,\ell]$.
Hence,
\[
\left| F_a(x)-F_a(y)\right|\le L|x-y|
\quad \text{for all $0\le x,y\le \ell$ and $a\in \tilde{K}$}
\]
Then, for every $a\in\tilde K$,
\[
F_a(x)\le F_a(y)+|F_a(x)-F_a(y)|
\le U''_\kappa(y)+L|x-y|.
\]
Taking the supremum over $a\in\tilde K$, we obtain
\[
U''_\kappa(x)\le U''_\kappa(y)+L|x-y|.
\]
Interchanging $x$ and $y$ yields
\[
U''_\kappa(y)\le U''_\kappa(x)+L|x-y|. 
\]
Therefore,
\[
|V_\kappa''(x)-V_\kappa''(y)|\le L|x-y|.
\]
Thus $D=V_\kappa''+V_\kappa'$ is also Lipschitz. Finally, for $x,y\in[0,\ell]$,
\[
|\beta(x)-\beta(y)|
=
\left|\frac{N(x)}{D(x)}-\frac{N(y)}{D(y)}\right|
\le
\frac{|N(x)-N(y)|}{m}
+\frac{\| N\|_\infty}{m^2}|D(x)-D(y)|.
\]
where $\| N\|_\infty:=\sup_{x\in [0,\ell]}|N(x)|$.
Since both $N$ and $D$ are Lipschitz, it follows that $\beta$ is Lipschitz on
$[0,\ell]$.
\end{proof}

\bibliographystyle{alpha}
\bibliography{sections/zotero-references}

\end{document}